\documentclass[11pt]{article}

\usepackage[left=1in,top=1in,right=1in,bottom=0.8in, nohead,nofoot]{geometry}
\def\asmall{0} 
\def\shownotes{0}  
\def\conference{0}

\usepackage{wrapfig,cite}
\usepackage[dvips]{graphicx}%
\usepackage{times,floatflt}
\usepackage{multirow}
\usepackage{rotating}
\usepackage{fullpage}
\usepackage[section]{placeins}
\usepackage{subfigure}
\usepackage{graphicx}
\usepackage{amsmath}
\usepackage{mathrsfs}
\usepackage{amssymb}
\usepackage{amsthm}
\usepackage{url}
\usepackage{caption, xspace}

\usepackage{pstricks,pst-plot}
\usepackage{pst-tree}

\usepackage{indentfirst}
\usepackage{clrscode}
\usepackage{array}
\usepackage{booktabs}
\newcommand{\ra}[1]{\renewcommand{\arraystretch}{#1}}

\usepackage{enumitem}
\setlist{nolistsep}

\usepackage[compact]{titlesec}

\newtheorem{claim}{Claim}[section]%
\newtheorem{theorem}{Theorem}[section]

\newtheorem{lemma}[theorem]{Lemma}
\newtheorem{proposition}[theorem]{Proposition}
\newtheorem{corollary}[theorem]{Corollary}
\newtheorem{definition}[theorem]{Definition}

\newcommand{\opt}{\mathrm{opt}}

\newcommand{\E}{\mathrm{E}}

\newcommand{\border}{\mathrm{border}}
\newcommand{\sink}{\mathrm{sk}}

\newcommand{\me}{\mathcal{E}}

\providecommand{\vs}{vs. }
\providecommand{\ie}{\emph{i.e.,} }
\providecommand{\eg}{\emph{e.g.,} }

\providecommand{\etal}{\emph{et al.}\xspace}
\providecommand{\etc}{\emph{etc.}}      
\providecommand{\mypara}[1]{\smallskip\noindent\emph{#1} }
\providecommand{\myparab}[1]{\smallskip\noindent\textbf{#1} }

\ifnum\shownotes=1
\newcommand{\authnote}[2]{{ $\ll$\textsf{\footnotesize #1 notes: #2}$\gg$}}
\else
\newcommand{\authnote}[2]{}
\fi
\newcommand{\Snote}[1]{{\authnote{Sharon}{#1}}}
\newcommand{\Znote}[1]{{\authnote{Zhenming}{#1}}}

\begin{document}

\title{The Diffusion of Networking Technologies}
\author{Sharon Goldberg\\Boston University \\{\tt goldbe@cs.bu.edu}
\thanks{Supported by NSF grant S-1017907 and a gift from Cisco.}
\and \quad
Zhenming Liu\\Princeton University \\{\tt zhenming@cs.princeton.edu}  \thanks{Supported by NSF grants CCF-0915922 and IIS-0964473.}
}

\date{\today}

\setcounter{page}{0}

\maketitle
\begin{abstract}
There has been significant interest in the networking community on the impact of cascade effects on the diffusion of networking technology upgrades in the Internet. Thinking of the global Internet as a graph, where each node represents an economically-motivated Internet Service Provider (ISP), a key problem is to determine the smallest set of nodes that can trigger a cascade that causes every other node in the graph to adopt the protocol. We design the first approximation algorithm with a provable performance guarantee for this problem, in a model that captures the following key issue: a node's decision to upgrade should be influenced by the decisions of the remote nodes it wishes to communicate with.

Given an internetwork $G(V,E)$ and threshold function $\theta$, we assume that node $u$ \emph{activates} (upgrades to the new technology) when it is adjacent to a \emph{connected component} of active nodes in $G$ of size exceeding node $u$'s threshold $\theta(u)$. Our objective is to choose the smallest set of nodes that can cause the rest of the graph to activate. Our main contribution is an approximation algorithm based on linear programming, which we complement with computational hardness results and a near-optimum integrality gap. Our algorithm, which does \emph{not} rely on submodular optimization techniques, also highlights the substantial algorithmic difference between our problem and similar questions studied in the context of social networks.

\myparab{Keywords.} Linear programming, approximation algorithms, diffusion processes, networks.

\myparab{Bibliographic note.}  An extended abstract of this work appeared in SODA'13.  This is the full version.
\end{abstract}
\newpage

\newpage
\tableofcontents
\newpage

\section{Introduction}\label{sec:intro}

There has been significant interest in the networking community on the impact of cascade effects on the diffusion of technology upgrades in the Internet~\cite{tussle,adoptability,JenYannis,adopt,ozment,JSGHZ08,ECS08A,edelman,newRochPaper,conext07}. %
Thinking of the global Internet as a graph, where each node represents an independent, economically-motivated \emph{autonomous system} (AS), \eg AT\&T, Google, Telecom Italia, or Bank of America, a key problem is to determine the set of nodes that governments and regulatory groups should target as early adopters of the new technology, with the goal of triggering a cascade that causes more and more nodes to  voluntarily adopt the new technology \cite{adopt,adoptability,newRochPaper,conext07}.
Given the effort and expense required to target ASes as early adopters, a natural objective (that has appeared in both the networking literature~\cite{adopt,adoptability,JenYannis} and also that of viral marketing~\cite{DomRich,KKT}) is to find the smallest possible \emph{seedset} of early adopters that could drive a cascade of adoption; doing this would shed light on how best to manage the upgrade from insecure routing~\cite{BGPsurvey} to secure routing~\cite{SBGP,BGPsec}, or from IPv4 to IPv6~\cite{ipv6}, or the deployment of technology upgrades like QoS~\cite{mescal}, fault localization~\cite{euroFL}, and denial of service prevention~\cite{SIFF}.

Thus far, the literature has offered only heuristic solutions to the problem of the diffusion of networking technologies.
In this paper, we design the first approximation algorithm with a provable performance guarantee that optimizes the selection of early adopter nodes, in a model of that captures the following important property: the technologies we study only allow a pair of nodes to communicate if they have a \emph{path} between them consisting of nodes that also use the new technology~\cite{adopt,adoptability,BGPsec,ipv6-nanog,mescal,euroFL,SIFF}.

\myparab{Model.} Consider a graph $G(V, E)$ that represents the internetwork. We use the following progressive process to model the diffusion of a new technology: a node starts out as inactive (using an older version of the technology) and \emph{activates} (adopts the new, improved technology) once it obtains sufficient utility from the new technology. Once a node is active, it can never become inactive. To model the cost of technology deployment, the standard approach~\cite{gran,Schelling,KKT} is to associate a threshold $\theta(u)$ with each node $u$ that determines how large its utility should be before it is willing to activate. A node's utility depends on the \emph{size of the connected components of active nodes adjacent to $u$ in $G$}. Thus, node $u$ activates if the connected component containing $u$ in the subgraph induced in $G$ by  nodes $\{v: v\in V, \text{Node }v \text{ is active}\} \cup \{u\}$ has size at least $\theta(u)$. We study the following optimization problem:

{\em \begin{quote} Given $G$ and the threshold function $\theta: V \rightarrow \{2,...,|V|\}$, what is the smallest \emph{feasible seedset} $S \subseteq V$ such that if nodes in $S$ activate, then all remaining nodes in $V$ eventually activate?
\end{quote}}

\smallskip\noindent
This model of node utility captures two key ideas:
\begin{enumerate}
 \item the traditional notion of ``direct network externalities/effects'' from  economics~\cite{katz1985network,farrell1985standardization}, marketing~\cite{bass} and other areas~\cite{metcalfe},  that supposes an active node that is part of a network of $k$ active nodes has utility that scales with $k$, and
  \item the fact that we are interested in networking technologies that only allow a pair of active nodes $u,v \in G(V,E)$ to communicate if there is path of active nodes between them in $G$.
\end{enumerate}
Our model has much in common with the vast literature on diffusion of innovations, and especially the linear threshold model for diffusion in social networks, articulated by Kempe \etal~\cite{KKT} and extensively studied in many other works.  Indeed, the two models diverge only in the choice of the utility function; ours is non-local, while theirs depends the (weighted) sum of a node's  active \emph{neighbors} in $G$.  Meanwhile, the non-local nature of our utility function has much in common with the classic literature on ``direct network externalities/effects'' \cite{katz1985network,farrell1985standardization,bass,metcalfe} with the important difference that these classic models ignore the underlying graph structure, and instead assume that utility depends on only a \emph{count} of the active nodes. We shall now see that these differences have a substantial effect on our algorithmic results.

\subsection{Our results.}\label{sec:ourResults}

\noindent
Our main result is an approximation algorithm based on  linear programming that consists  of two phases.  The first is a linearization phase that exploits combinatorial properties to encode our problem as an integer program (IP) with a 2-approximate solution, while the second is a randomized rounding algorithm 
that operates by restricting our search space to \emph{connected seedsets}, \ie seedsets that induce a connected subgraph of $G$ .
%
We have:
\begin{theorem}[Main result]\label{thm:mainInformal} Consider a networking technology diffusion problem $\{G(V,E), \theta\}$ where the smallest seedset has size $\opt$, the graph has diameter $r$ (\ie $r$ is the length of ``longest shortest path'' in $G$), and there are at most $\ell$ possible threshold values, \ie $\theta:V\to \{\theta_1,...,\theta_\ell\}$. Then there is a polynomial time algorithm that returns a seedset $S$ of size
 $O( r\ell \log|V| \cdot \opt)$.
\end{theorem}

\myparab{Relationship to the linear threshold model in social networks. }
Our main result highlights the major algorithmic difference between our work and the linear threshold model in social networks~\cite{KKT}. In the social network setting, Chen~\cite{Chen08} showed that this problem is devastatingly hard, even when $r,\ell= O(1)$; 
to avoid this discouraging lower bound, variations of the problem that exploit submodular properties of the objective have been considered (\eg where thresholds are chosen uniformly at random~\cite{KKT} or see \cite{CCC11,GRS12} and references therein).  Indeed, the ubiquity of these techniques seems to suggest that diffusion problems are tractable \emph{only} when the objective exhibits submodularity properties. Our work provides an interesting counterpoint: our positive result does not rely on submodular optimization, 
and we show that the influence function in our problem, and its natural variations, lacks submodularity properties.

%
%

\myparab{Dependencies on $r$, $\ell$, and $\log |V|$ are necessary.}  Removing our algorithm's  dependence on $r,\ell$, or $\log|V|$ is likely to require a very different set of techniques
because of the following barriers:
\begin{enumerate}
\item \mypara{Computational barrier.}  We use a reduction from Set Cover to show that our problem does not admit any $o(\ln |V|)$-approximation algorithm, even if $r,\ell=O(1)$.

\item \mypara{Combinatorial barrier.}  We present a family of problem instances that prove that any algorithm that returns a connected seedset must pay an $\Omega(r)$-increase in the size of the seedset in the worst case.


\item \mypara{Integrality gap.}  The linear program we use has an integrality gap of $\Omega(\ell)$ so that our rounding algorithm is asymptotically optimal in $\ell$.

\end{enumerate}

\myparab{Quality of approximation. } We interpret the quality of our approximation for typical problem instances.

\mypara{Networking.}  The motivation for our problem is to help centralized authority (\eg a government, a regulatory group) determine the right set of autonomous systems (ASes) in the Internet to target as early adopters for an upgrade to a new networking technology \cite{BGPFCC,obamaIPv6}.  We comment on the asymptotic order of $r$ and $\ell$ when a centralized authority executes this algorithm.  The graph $G$ is the Internet's AS-level graph, which is growing over time, with diameter $r$ that does not exceed $O(\log |V|)$ (see, \eg~\cite{jureDiameter}).   We remark that the empirical data we have about the Internet's AS-level topology~\cite{cyclops,CAIDA,ixp,DIMES} is the result of a long line of Internet measurement research~\cite{tenlessons}.
On the other hand, obtaining empirical data on ASes' thresholds is still subject to ongoing research~\cite{quicksand,BGPFCC}. The following natural assumption and practical constraint restrict the threshold granularity $\ell$: (a) ASes should not be sensitive to small changes in utility (\eg 1000 nodes \vs 1001 nodes), and that (b) in practice, it is infeasible for a centralized authority to obtain information about $\theta(u)$ from every AS $u$ in the Internet, both because this business information is kept private and because, perhaps more importantly, many of these nodes are in distant and possibly uncooperative countries.
Thus, thresholds should be chosen from a geometric progression  $\{(1+\epsilon),(1+\epsilon)^2,....,(1+\epsilon)^\ell\}$ or even restricted to a constant size set $\{5\%,10\%,15\%,20\%,30\%,50\%\}$ as in \cite{adopt,adoptability,ozment} so that $\ell=O(\log |V|)$.  Our approximation ratio is therefore polylogarithmic in $|V|$ in this context.

\mypara{Other settings. }  Since our model is a general, there could be other settings where
$\ell$ may not be $O(\log|V|)$. Here, the performance of our algorithms is governed by the \emph{stability} of the problem instance. Stability refers to the magnitude of the change in the optimal objective value due to a perturbation of the problem instance, and is commonly quantified using \emph{condition numbers} (as in \eg numerical analysis and optimization~\cite{Demmel87,HJ99,LL10}). We naturally expect unstable problem instances (\ie with large $\kappa$) to be more difficult to solve.  Indeed, we can use condition numbers to parameterize our approximation ratio:
%
%
\begin{definition}[Condition number]\label{def:condition}  Consider a problem instance $\Pi = \{G, \theta\}$ and a positive constant $\epsilon$. Let $\Pi^+ = \{G, \theta^+\}$ and $\Pi^- = \{G, \theta^-\}$ be two problem instances on the same graph $G$ where for every $v\in V$, we have $\theta^+(v) = (1+\epsilon)\theta(v)$ and $\theta^-(v) = (1-\epsilon)\theta(v)$. Let $\opt^+$ ($\opt^-$) be the value of the optimal solution for $\Pi^+$ ($\Pi^-$). The \emph{condition number} is $\kappa(\Pi, \epsilon) \triangleq \frac{\opt^+}{\opt^-}.$
\end{definition}
\begin{corollary}\label{cor:little}Let $\epsilon$ be an arbitrary small constant. There exists an efficient algorithm to solve a technology diffusion problem $\Pi = \{G, \theta\}$ whose approximation ratio is $\tilde O(\kappa(\Pi, \epsilon)\cdot r)$.
\end{corollary}
See details in Appendix~\ref{apx:little}.

\smallskip
Finally, we remark that our IP formulation might also be a promising starting point for the design of new heuristics.  Indeed, in
\ifnum \conference=0
Appendix~\ref{sec:exp}
\else
Section~\ref{sec:exp}
\fi
 we ran a generic IP solver to find seedsets on problem instances of non-trivial size; the seedsets we found were often substantially better than those returned by several natural heuristics (including those used in~\cite{adopt,adoptability,JenYannis}).


\ifnum \conference=0
\myparab{Organization.} We present our IP formulation in Section~\ref{sec:linearize}, and describe our rounding algorithm in Section~\ref{sec:roundingAlgo}. All missing proofs are in Appendix~\ref{apx:connectSeq}-\ref{apx:little}. Lower bounds are in Appendix~\ref{apx:hard}.  We also present supplementary material on the (lack of) submodularity/supermodularity properties of our problem (Appendix~\ref{sec:notSubmod}), our experimental results (Appendix~\ref{sec:exp}), and expository examples and figures (Appendix~\ref{apx:example}).
\else
\myparab{Organization.} We present our IP formulation in Section~\ref{sec:linearize} and describe our rounding algorithm in Section~\ref{sec:roundingAlgo}. Our experimental results will
be presented in Section~\ref{sec:exp}.
All missing proofs are in Appendix~\ref{apx:connectSeq}-\ref{apx:little}. Expository
examples are in Appendix~\ref{apx:example}.
The full version of this paper also presents the lower bounds and supplementary material on
(lack of) submodularity/supermodularity properties of our problem.
\fi

\section{\textbf{Linearization \& formulating the IP}}\label{sec:linearize}

In this section we show how to sidestep any potential difficulties that could result from the  non-local nature of our setting. To do this, we restrict our problem in a manner that allows for easy encoding using only linear constraints, while still providing a 2-approximation to our objective.
We need the following notions:

\myparab{Activation sequences. }
Given a seedset $S$, we can define an \emph{activation sequence} $T$ as a permutation from $V$ to $\{1,...,n\}$ where $n=|V|$ that indicates the order in which nodes activate. The $t$-th position in the sequence is referred as the $t$-th timestep.
We allow a seed node to activate at \emph{any} timestep, while a non-seed node $u$ may activate at a timestep $T(u)$ as long as $u$ is part of a connected component of size at least $\theta(u)$ in the subgraph of $G$ induced by $\{u\}\cup\{v: T(v) < T(u)\}$.

\myparab{Connected activation sequences. } A connected activation sequence $T$ is an activation sequence such that at every timestep $t$, the set of active nodes induces a connected subgraph of $G$. %
We may think of $T$ as a spanning tree over the nodes in the graph, where, at every timestep, we add a new node $u$ to the tree subject to the constraint that $u$ has a neighbor that is already part of the tree.

Our IP will find the smallest seedset $S$ that can induce a connected activation sequence. At first glance this could result in a factor of $r$ growth in the seedset size.  However, the following lemma, which may be of independent interest, shows that the seedset size grows at a much smaller rate:

\begin{lemma}\label{lem:connectSequence}
The smallest seedset that can induce a connected activation sequence is at most twice the size of the optimal seedset.
\end{lemma}

\begin{proof}[Sketch of proof of Lemma~\ref{lem:connectSequence}] (The full proof is in Appendix~\ref{apx:connectSeq}.) We prove that any activation sequence $T^*$ induced by the optimal seedset $\opt$ can be rearranged to form a \emph{connected activation sequence} $T$ if we add at most $|\opt|$ extra nodes to the seedset. To see how, consider a timestep in $T^*$ when two or more connected active components merge into a single component, and notice that whenever this happens, there is exactly one \emph{connector} node that activates and joins these two components. In the full proof we show that by adding every connector to the seedset, we can rearrange $T^*$ to obtain a connected activation sequence.  It remains to bound the number of connectors.  Since every connector node decreases the number of disjoint connected components, and each component must contain at least one seed, then there is at most one connector for each seed node, and the 2-approximation follows.
\end{proof}


\ifnum\conference=0
\vspace{-5pt}
\ifnum\asmall=1\begin{figure}[b]
\else\begin{figure}[htp]\fi
\vspace{1.5cm}
\else\begin{figure*}[htp]
\vspace{1.5cm}
\fi
\centering
\ifnum\conference=0
\begin{pspicture}(0,-1.5)
\else
\begin{pspicture}(0,-2)
\fi
\psframebox{
\small
$\begin{array}{llll}
  &\min& \sum_{i \leq n}\sum_{t < \theta(v_i)}x_{i,t} &\\
\mbox{subject to:} & \forall t,i: &  x_{i, t} \in \{0, 1\} & \\
&\forall i: & \sum_{t \leq n}x_{i, t}  =   1  &\mbox{(permutation constraints)}\\
&\forall t:  & \sum_{i \leq n}x_{i, t}   =   1 & \mbox{(permutation constraints)}\\
& \forall t>1,i: & \sum_{\{v_{i}, v_{i'} \}\in E}\sum_{t' < t}x_{i',t'}  \geq  x_{i,t} & \mbox{(connectivity constraints)}
\end{array}$
}
\end{pspicture}
\vspace{-3mm}
\ifnum\conference=1 \vspace{-7mm} \fi 
\caption{Simple IP for the networking technology diffusion problem.}
\label{fig:ip1}
\ifnum\conference=0
\end{figure}
\else
\vspace{-3mm}\end{figure*}
\fi

\myparab{IP encoding. } The beauty of a connected activation sequence $T$ is that every nonseed node's decision to activate becomes local, rather than global: node $v$ need only check if (a) at least one of its neighbors are active, and (b) the current timestep $t$ satisfies $t\geq\theta(v)$. Moreover, given a connected activation sequence $T$, we can uniquely recover the smallest feasible seedset $S$ that could induce $T$ by deciding that node $u$ is a seed iff $\theta(v)>T(v)$.  Thus, our IP encodes a connected activation sequence $T$, as a proxy for the seedset $S$.
Let $\{v_1, v_2, ..., v_n\}$ be the set of nodes in the network.
Let $x_{i,t}$ be an indicator variable such that $x_{i, t} = 1$ if and only if $T(v_i) = t$. The integer program is presented in Figure~\ref{fig:ip1}.  The \emph{permutation constraints} guarantee that the variables $x_{i, t}$ represent a permutation. The \emph{connectivity constraints} ensure that if $x_{i, t} = 1$ (\ie node $v_i$ activates at step $t$), there is some other node $v_{i'}$ such that $v_{i'}$ (a) is a neighbor of node $v_i$ and and (b) activates at earlier time $t' < t$. Finally, the objective function minimizes the size of the seedset by counting the number of $x_{i, t} = 1$ such that $t < \theta(v_i)$.

We remark that our IP formulation suggests a similarity between our setting and the vehicle routing with time windows problem (~\eg~\cite{desrochers1987vehicle,BCC94,bansal2004approximation,FredericksonW12}).
Consider a time-windows problem, where we are given an undirected metric graph $G$ and time window $[r(u),d(u)]$ for each node $u$, and our objective is to choose a tour for the vehicle through $G$  that visits as many nodes as possible during their respective time windows.  In our setting (restricted to connected activation sequences), the tour becomes a spanning tree, and each node $u$ has time window $[\theta(u),n]$.   Understanding the deeper connection here is an interesting open question.

\section{Rounding algorithm.}\label{sec:roundingAlgo}

Unfortunately, the simple IP of Figure~\ref{fig:ip1} has a devastating $\Omega(n)$ integrality gap
\ifnum\conference=0
(Appendix~\ref{sec:gap1}).
\else
(See the full paper).
\fi
We eliminate this integrality gap by adding extra constraints to the IP of Figure~\ref{fig:ip1}, and refer to the resulting IP as the \emph{augmented IP}. We defer presentation of this IP to Section~\ref{sec:flowIP} and focus now on the high level structure of our rounding algorithm.   

Our rounding algorithm is designed to exploit the relationship between seedset $S$ and connected activation sequences $T$; namely, the fact that we can uniquely recover a $S$ from $T$ by deciding that node $u$ is a seed if $T(u)<\theta(u)$.  As such, it returns \emph{both}  $S$ and $T$ with the following four properties:
\begin{enumerate}
\item\label{prop:consist} \emph{Consistency.}  $S$ and $T$ are \emph{consistent}; namely, $T$ is an activation sequence for the diffusion process induced by $\{G, \theta, S\}$.  (Recall that $T$ is such that any seed $u \in S$ can activate at any time, and any non-seed $u \notin S$ can activate whenever it is connected to an active component of size at least $\theta(u) - 1$.)
\item\label{prop:feas} \emph{Feasibility.} $T$ is such that every node eventually activates.
\item\label{prop:conn} \emph{Connectivity.} $T$ is a connected activation sequence.
\item\label{prop:small} \emph{Small seedset.} The seedset $S$ has ``small'' size, \ie  bounded in  the size of the objective function of the solution to our LP.
\end{enumerate}

\smallskip\noindent
But how should we round the fractional $x_{i,t}$ values returned our LP relaxation to achieve this? Let's first consider two natural approaches for sampling $S$ and $T$:

\myparab{Approach 1: Sample the seedset $S$:}  Recall that in a connected activation sequence, a node that activates at time $t<\theta(u)$ must be a seed. Therefore, we can sample the seedset $S$ by adding each node $v_i$ to $S$ with probability proportional to $\sum_{t<\theta(v_i)}x_{i,t}$.

\myparab{Approach 2: Sample the activation sequence $T$:}  We can instead sample the activation sequence $T$ by deciding that node $v_i$ activates before time $t$ with probability proportional to $\sum_{\tau < t} x_{i,\tau}$.

\smallskip\noindent
However, neither of these approaches will work very well. While Approach 1 guarantees that the seedset $S$ is small (Property~\ref{prop:small}), it completely ignores the more fine-grained information provided by the $x_{i,t}$ for $t\geq\theta(v_i)$ and so its not clear that nonseed nodes will activate at the right time (Property~\ref{prop:feas}).  Meanwhile, Approach 2 guarantees feasibility (Property~\ref{prop:feas}), but by sampling activation times for each node independently, it ignores correlations between the $x_{i,t}$. It is therefore unlikely that the resulting $T$ is connected (Property~\ref{prop:conn}), and we can no longer extract a small seedset (Property~\ref{prop:small}) by checking if $T(u)<\theta(u)$.

Instead, we design a sampling procedure that gives us a coupled pair $\{S, T\}$ where, with high probability, (a) the distribution of $S$ will be similar to that of Approach 1, so that the seedset is small (Property~\ref{prop:small}), while (b) the distribution of $T$ will be similar to Approach 2, so we have feasibility (Property~\ref{prop:feas}), and also (c) that $T$ is connected (Property~\ref{prop:conn}).  However, $S$ and $T$ are not necessarily consistent (Property~\ref{prop:consist}). Later, we show how we use repeated applications of the sampling approach below to correct inconsistencies, but for now, we start by presenting the sampling routine:

\myparab{Approach 3: Coupled sampling. }  We start as in Approach 1, adding each $v\in V$ to $S$ with probability
$\min\left\{1, \alpha\sum_{t < \theta(v)}x_{v,t}\right\}$,
%
where $\alpha$ is a bias parameter to be determined.
We next run \emph{deterministic} processes: $S \gets \proc{Glue}(S)$ followed by $T\gets \proc{GetSeq}(S)$.

The $\proc{Glue}$ procedure, defined below, ensures that $S$ is connected (\ie induces a connected subgraph of $G$), and blows up $|S|$ by an $O(r)$-factor (where $r$ is graph diameter). Meanwhile, $\proc{Get-Seq}$, defined in Section~\ref{sec:whyGetSeqWorks}, returns a connected activation sequence $T$; we remark that $T$ may not be a permutation (many nodes or none could activate in a single timestep), and may not be feasible, \ie activate every node.

\ifnum\conference=1
\begin{figure*}
\fi
\begin{codebox}
\Procname{$\proc{Glue}(S)$}
\li \While $S$ is not connected
\li \Do Let $C$ be the largest connected component in the subgraph induced by $S$.
\li Pick $u \in S \setminus C$. Let $P$ be the shortest path connecting $u$ and $C$ in $G$.
\li Add nodes in $P$ to $S$.
\End
\li \Return S.
\end{codebox}
\ifnum\conference=1
\vspace{-7mm}
\caption{The $\proc{Glue}$ procedure for constructing a connected seedset.}
\vspace{-3mm}
\end{figure*}
\fi

The properties of Approach 3 are captured formally by the following proposition, whose proof (Section~\ref{sec:mainprop}) presents a major technical contribution of our work:

\begin{proposition}\label{prop:flow} Let $\alpha = 24(1+\epsilon)\ln(\frac{4n^2}{\epsilon})$ and $\epsilon$ be a suitable constant. Then there exists an augmented IP and an efficiently computable function $\proc{Get-Seq}(\cdot)$ such that Approach 3 returns $S$ and $T$ (that are not necessarily consistent), where
\begin{enumerate}
\item $T$ is connected,
\item for any $v \notin S$ we have that $T(v) \geq \theta(v)$, and
\item for any $v$ and $t$,
\begin{itemize}
\item if $\sum_{t' \leq t}x_{v,t'}\geq \frac{1}{12(1+\epsilon)}$, then $\Pr[T(v) \leq t] \geq 1 - \frac{\epsilon}{4n^2}$.
\item if $\sum_{t' \leq t}x_{v,t'}< \frac{1}{12(1+\epsilon)}$, then $\Pr[T(v) \leq t] \geq (1+\epsilon)(\sum_{t' \leq t}x_{v,t'})$.
\end{itemize}
\end{enumerate}
\end{proposition}

Notice that the third item in Proposition~\ref{prop:flow} suggests that the distribution of $T$ in Approach 3 is ``close'' to that of Approach 2. We also remark that in this item the parameters are not optimized. Thus when we are right at the transition point (\ie $\sum_{t' \leq t}x_{v,t'}= \frac{1}{12(1+\epsilon)}$),  the two cases give quite different bounds.
 In the subsequent section, we apply the ideas we developed thus far to design an algorithm that uses Proposition~\ref{prop:flow} to ``error-correct'' inconsistencies between $S$ and $T$ so that all four properties are satisfied.   Then, in Section~\ref{sec:mainprop} we present the more technically-involved proof of Proposition~\ref{prop:flow}.


\subsection{Resolving inconsistencies via rejection-sampling
\ifnum\conference=1\\\fi}
\label{sec:roundAlgoOverview}

\ifnum\conference=1$\;$\vspace{-3mm}\fi

Recall the threshold function $\theta: V \to \{\theta_1,...,\theta_\ell\}$, and suppose a threshold $\theta_j$ is \emph{good} with respect to $T$ if there are at least $\theta_j-1$ active nodes in $T$ by time $\theta_j - 1$.   The following simple lemma presents the properties we need from our rejection sampling algorithm:

\begin{lemma}\label{lem:goodst} Let $S$ be a seedset and $T$ be an activation sequence. If:
\begin{itemize}
\item (P1). $T$ is connected and feasible (for any $v \in V$, $T(v) \leq n$), and
\item (P2). $T(v) \geq \theta(v)$ for all $v \notin S$, and
\item (P3). Every $\theta_j$ for $j \in [\ell]$ is good with respect to $T(\cdot)$,
\end{itemize}
then $S$ is consistent with $T$ and $S$ is a feasible seedset.
\end{lemma}
\begin{proof} 
To show that $S$ and $T$ are consistent, we argue that by the time a non-seed $v\notin S$ activates in $T$, there are at least $\theta(v) - 1$ active nodes.  Since $v$ activates at time $T(v)\geq \theta(v)$, this follows because $T(\cdot)$ is connected and each $\theta_j$ is good. Since $T$ is feasible and $S$ is consistent with $T$, we have that $S$ is feasible.
\end{proof}

We construct a pair of $\{S, T\}$ that meets the properties of Lemma~\ref{lem:goodst} in two phases. First, we construct $\ell$ pairs $\{S_1, T_1\}$, ..., $\{S_{\ell}, T_{\ell}\}$ where for each $\{S_j, T_j\}$  we have that (P.1) and (P.2) hold, a \emph{single} threshold $\theta_i$ is good w.r.t. $T_i$, and $S_i$ is ``small'', \ie
$|S_i| \leq 24(1+\epsilon)^2\ln(\frac{4n^2}{\epsilon})r\cdot \opt$.
The second phase assembles these $\ell$ pairs into a single $\{S, T\}$ pair so \emph{all} $\theta_j$ are good w.r.t. $T$, so that (P1)-(P3) hold, and the seedset $S$ is bounded by $O(r \ell \ln n \cdot \opt)$, so our main result follows.

\myparab{Step 1. Rejection-sampling to find $\{S_j, T_j\}$ pairs $\forall j \in [\ell]$.} Thus, while we sample $T_j$ that may not be permutations, the following lemma (proved in Appendix~\ref{apx:proofLemSample}) shows that we can repeat Approach 3 until we find $S_j,T_j$ that satisfy the required properties:

\begin{lemma}[Success of a single trial]\label{lem:sample} Let $S_j$ and $T_j$ be sampled as in Approach 3. For any $t$, let $A_t$ be the number of nodes active in $T_j$ up to time $t$ (inclusive). Then
$\Pr[A_t \geq t \wedge A_n = n] \geq \frac{\epsilon}{2n}$.
\end{lemma}
\noindent To see why, observe that (P1)-(P2) hold by Proposition~\ref{prop:flow}, and $\theta_j$ is good w.r.t. $T_j$ with probability $\frac{\epsilon}{2n}$ by Lemma~\ref{lem:sample}, and $S_j$ has the required size with probability $\geq 1 - \frac 1 {n^{10}}$ by standard Chernoff bounds (the exponent 10 here is chosen arbitrarily). Therefore, we successfully find the required $\{S_j, T_j\}$ with probability $\frac{\epsilon} {2n} - \frac{1}{n^{10}}$ in a single trial.  After $O(n \log n)$ independent trials, we find the required $\{S_j, T_j\}$ with probability $1-1/{n^c}$ for sufficiently large $c$.

\myparab{Step 2. Combine the $\{S_i, T_i\}$ to obtain the final $\{S, T\}$.} We can now construct our final $\{S, T\}$ pair in a rather straightforward way: to construct $S$, we take the union of all the $S_j$'s and then use $\proc{Glue}$ to connect them; that is we take $S \gets \proc{Glue}(\bigcup_{j\leq \ell} S_j)$. To construct $T$, we set $T(v)=1$ for all seeds $v\in S$ and $T(v) = \min_{j \leq \ell}T_j(v)$ ($\forall v\in V\backslash S$).

\noindent To conclude, we need show that this $\{S, T\}$ pair satisfies Lemma~\ref{lem:goodst}.
\begin{itemize}
\item First we show (P1) holds. Since every $T_j$ is feasible, and $T(v) \leq T_j(v)$ by construction, it follows that $T$ is also feasible. Next we show that $T(v)$ is connected by induction over $t$.  As a base case, observe that $T = \min_{j \leq \ell}T_j(v)$ is connected at $t=1$, since the seedset $S = \proc{Glue}(\bigcup_j S_j)$ is connected.  As the induction step, we assume that $T$ is connected up to time $t$ (inclusive) and show that $T$ is also connected up to time $t+1$ (inclusive).  To do this, let $v$ be a node such that $T(v) = t+1$.  It follows that there exists $j\leq\ell$ such that $T_j(v)=t+1$; since $T_j$ is connected, there must be another node $u$ such that $T_j(u) < t+1$ and $u$ and $v$ are neighbors in the graph $G$.  Since $T(u)\leq T_j(u)$, it follows that $v$ is connected to a node (namely node $u$) that is active at time $t + 1$, and the induction step follows.

\item We show that (P2) holds. For all $v \notin S$, we have $v \notin S_j$ for all $j \leq \ell$. This means $T_j(v) \geq \theta(v)$ for all $j$. Therefore, $T(v) \geq \theta(v)$ and (P2) holds.
\item Finally, (P3) holds. For each $j\leq \ell$ we know that $\theta_j$ is good w.r.t to $T_j$.  For all $j\leq \ell$, every node $v$ has $T(v) \leq T_j(v)$ by construction, so that the number of active nodes at time $\theta_j$ in $T$ must be no fewer than the number of active nodes in $T_j$.  (P3) follows since $\theta_j$ is good w.r.t to $T_j$ for every $j\leq \ell$.
\end{itemize}
It follows that Lemma~\ref{lem:goodst} holds and the final seedset $S$ is indeed a feasible seedset.  Since the size of each seedset $S_i$ is bounded by $O(r\log n\cdot\opt )$ (and the gluing in Phase 2 grows the seedset by an additive factor of at most $\ell\cdot r$) it follows that  $S$ has size at most $O(\ell r \log n \cdot \opt )$ and our main result follows.

\subsection{Strengthened IP and coupled sampling}\label{sec:mainprop}
We show how we use a flow interpretation of our problem to prove Proposition~\ref{prop:flow}.

\subsubsection{The need for stronger constraints
\ifnum\conference=1\\\fi}\label{sec:badexample}

\ifnum\conference=1$\;$\vspace{-3mm}\fi

In
\ifnum \conference=0
Appendix~\ref{sec:gap1},
\else
the full paper,
\fi
we show that the LP in Figure~\ref{fig:ip1} has an $\Omega(n)$ integrality gap.  To understand why this gap comes about, let us suppose that each $x_{i, t}$ returned by the LP is a mass that gives a measure of the probability that node $v_i$ activates at time $t$. Consider the following example:

\ifnum \conference=0
\begin{wrapfigure}{r}{10cm}
 \vspace{-12pt}\includegraphics{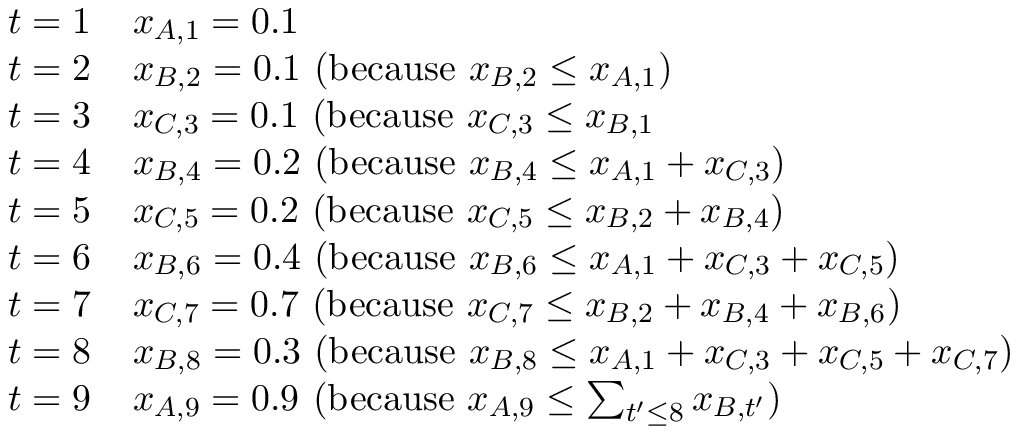} \vspace{-30pt}
\end{wrapfigure}
\else
\begin{figure*}
\begin{center}
\includegraphics{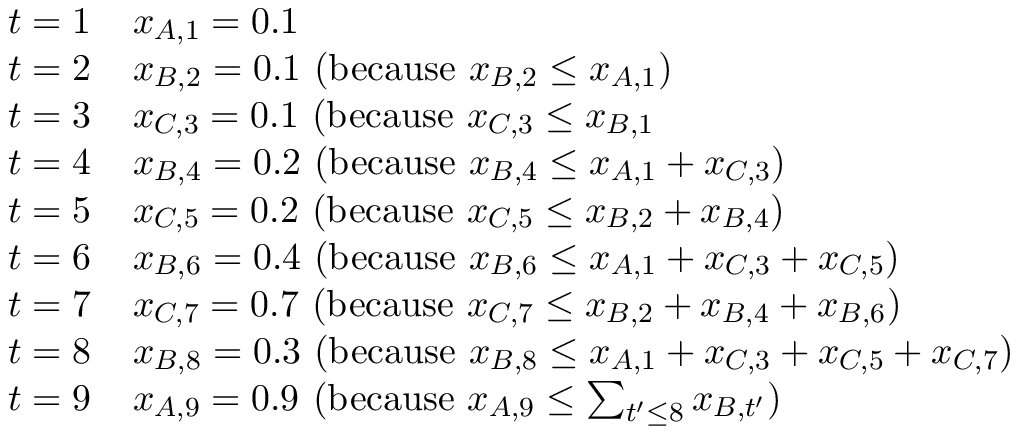} 
\end{center}
\vspace{-7mm}
\caption{A pathological example for the original LP formulation.}\label{fig:pathosmall}
\vspace{-3mm}
\end{figure*}
\fi

\myparab{Pathological example.}
Consider a graph that contains a clique of nodes A,B and C. Suppose the LP returns a solution such that at $t = 1$, node A has mass 0.1, while all other nodes have mass 0. The constraints 
\ifnum\conference=1
in Figure~\ref{fig:pathosmall} 
\fi
 repeatedly allow mass from node $A$ to circulate through nodes $B$ and $C$ and then back to $A$, as shown in the variable assignments beside. Finally, at $t=9$, enough mass has circulated back to $A$, so that $A$ has mass 0.9 and thus ``probability'' 0.9 of activating.  Note that this is highly artificial, as all of this mass originated at $A$ to begin with!  In fact, no matter how we interpret these $x_{i, t}$, the example suggests that this ``recirculation of mass'' is unlikely to give us any useful information about when node $A$ should actually activate.

\ifnum\conference=1 In the full version, we show how to use this ``recirculation of mass'' to construct an $\Omega(n)$-integrality gap for the LP in Figure~\ref{fig:ip1}.
\fi

\subsubsection{The flow constraints
\ifnum\conference=1\\\fi}\label{sec:flowIP}

\ifnum\conference=1$\;$\vspace{-3mm}\fi

Inspired by the example above, we can reduce the size of the integrality gap by thinking of the diffusion process in the context of network flows.  Specifically, we suppose that when a nonseed node $u$ activates at time $T(u)$, a unit flow originates at a seed node and flows to node $u$ along the network induced by the nodes active prior to timestep $T(u)$. We therefore augment the IP of Figure~\ref{fig:ip1} with this idea by introducing \emph{flow constraints}:

\myparab{The flow network.}
For any solution $\{x_{i, t}\}_{i, t \leq n}$, we define a \emph{flow network} $\mathcal H$, with vertex set $V(\mathcal H) = \{X_{i, t}: i , t \in [n]\}$ and edge set $E(\mathcal H)=\{(X_{i, t}, X_{i',t'}): t' > t \wedge \{v_{i'}, v_i\} \in E(G)\}$.
Every node $X_{i,t}$ in the flow network $\mathcal H$ has capacity $x_{i,t}$, while edges in $\mathcal H$ do not have capacity bounds.
We let the line that connects nodes $\{X_{i,t}: t = \theta(v_i)\}$ be the \emph{threshold line}. All the $X_{i,t}$ such that $t < \theta(v_i)$ are flow graph nodes \emph{to the left of the threshold line}; very roughly, these nodes corresponds the region where $v_i$ is a seed. The rest are flow graph nodes \emph{to the right of the threshold line}, and roughly correspond to $v_i$ being a nonseed.  A sample flow graph and its threshold line
appears in Appendix~\ref{apx:example}.

\myparab{Flow constraints. } For now, we suppose the first node to be activated in the optimal solution is known to be $v_1$ (so that $x_{1,1} = 1$); this assumption is removed in Appendix~\ref{apx:flowIP}. For any $i$ and $t \geq \theta(v_i)$, we define the $(i, t)$-flow as the multiple-sink flow problem over the flow network $\mathcal H$, where  the source is $X_{1,1}$ and the sinks are nodes to the right of threshold line, namely $\{ X_{i, \theta(v_i)}$, $X_{i, \theta(v_i) + 1}$, ..., $X_{i, t}\}$. The demand for the sink $X_{i,t}$ is $x_{i,t}$.
Our flow constraints require that every $(i, t)$-flow problem (for all $i$ and all $t \geq \theta(v_i)$) has a solution.
In Appendix~\ref{apx:flowIP} we show that how to implement these flow constraints using the maximum-flow-minimum-cut theorem and a separation oracle. Appendix~\ref{apx:flowIP} presents the implementation of the augmented IP in  Figure~\ref{fig:ipflow}, as well as the proof of the following:

\begin{lemma}\label{lem:flowip} The augmented IP for the technology diffusion problem is such that
\begin{itemize}
\item when $T(v_1) = 1$ in the optimal connected activation sequence, this IP returns the same set of feasible solutions as the simple IP of Figure~\ref{fig:ip1}.
\item the fractional solution for the corresponding relaxed LP satisfies all the $(i, t)$-flow constraints.
\end{itemize}
\end{lemma}

\myparab{Eliminating the integrality gap. }
The flow constraints eliminate the pathological example above, and therefore also the $\Omega(n)$ integrality gap. To see why, notice that the $(B, 4)$-flow problem has total demand 0.2 (\ie $x_{B,2}=0.1$ and $x_{B,4}=0.1$) but there is no way to supply this demand from $X_{A,1}$.

\subsubsection{Why coupled sampling works.
\ifnum\conference=1\\\fi }\label{sec:whyGetSeqWorks}

\ifnum\conference=1$\;$\vspace{-3mm}\fi

In addition to improving the robustness of our IP, the flow constraints also have the following pleasant interpretation that we use in the design our rounding algorithm: if there is a flow $f \in [0, 1]$ from a seed node to a non-seed node $u$ at time $t$, then node $u$ has probability $f$ of activating at time $t$.

\myparab{On connected seedsets. } To ensure that all network flows originate at seed nodes, Approach 3 requires $\proc{Get-Seq}$ to return an activation sequence $T$ where all seed nodes activate before the non-seed nodes.  If we couple this with the requirement that $T$ is connected (so we can use the trick of deciding that node $v$ is a seed if $T(v)<\theta(v)$), it follows that we require a \emph{connected seedset} $S$ (\ie the nodes in $S$ induce a connected subgraph of $G$). Approach 3 achieves this by using $\proc{Glue}$ to connect the nodes it samples into its seedset $S$, and then deterministically generates $T$ using $\proc{Get-Seq}$ as specified 
\ifnum\conference=1
in Figure~\ref{fig:getseq},
\else
below,
\fi
and illustrated in Appendix~\ref{apx:example}.

\ifnum \conference=1
\begin{figure*}
\fi
\begin{codebox}
\Procname{$\proc{Get-Seq}(\mathcal H, S)$}
\li Initialize by flagging each $X_{u,t} \in \mathcal H$ as ``inactive'' by setting $b_{u,t}\gets 0$.
\li $\forall\;u \in S$, $b_{u,t} \gets 1$ for all $t$. $\qquad$\slash\slash\;``Activate'' all $X_{u,t}$ for all $u$ in the seedset. 
\li \For $t \gets 1$ to $n$
\li \Do $\forall \; u$ s.t. $\theta(u) \geq t$:
\li \If
$\left(\exists v, \tau: ((X_{v,\tau}, X_{i,t}) \in E(\mathcal H)) \wedge (b_{v,\tau} = 1) \right)$
\li  \;\;$b_{u,t'} \gets 1$ for $t' \geq t$ $\qquad$\slash\slash\;``Activate'' each $X_{u,t'}$ to the right of timestep $t$. 
\End
\li Obtain $T$ by taking $T(u)\gets \min \left\{t: b_{u,t}=1\right\}$ for every $u \in V$.
\li \Return $T$.
\end{codebox}
\ifnum \conference=1
\vspace{-7mm}
\caption{The $\proc{Get-Seq}$ procedure for constructing the function $T(\cdot)$.}\label{fig:getseq}
\vspace{-3mm}
\end{figure*}
\fi

\myparab{Intuition behind the proof of Proposition~\ref{prop:flow}.} Given the probabilistic interpretation of flows, consider what happens if two disjoint flows $f_1$ and $f_2$ originate from different seeds and arrive simultaneously at node $u$ at time $t$. The total flow at node $u$ at time $t$ is then $f_1 + f_2$. What does this merge of two disjoint flows mean in our probabilistic interpretation?
It turns out that the natural interpretation is already pretty sensible: with probability $f_1$, the technology is diffused via the first flow, and with probability $f_2$ the technology is diffused via the second flow. Now, the probability that the technology is diffused to $u$ via either of these two flows is $1-(1-f_1)(1-f_2)$.  When $f_1,f_2$ are both small, this probability becomes $\approx f_1+ f_2$, so that the total flow can be used to determine node $u$'s activation probability.
%
On the other hand, when $f_1$ or $f_2$ is large, we are fairly confident that $u$ should activate prior to time $t$, and so we can simply decide that $T(u) \leq t$ without incurring a large increase in the size of the seedset. Given that the total demand in the $(u,t)$ flow problem is $\sum_{\theta(u)\leq\tau\leq t} x_{u,\tau}$, it follows that the probability that $u$ is a nonseed and is activated by time $t$ is roughly proportional to this demand. Also, notice that $u$ itself is chosen as a seed with probability $\sum_{\tau<\theta(u)} x_{u,\tau}$ so by combining these events in the appropriate way, we get that $\Pr[T(u) \leq t] \propto \sum_{\tau<t} x_{u,\tau}$ as required by the third item in
Proposition~\ref{prop:flow}.

To formalize this intuition, we first describe how \proc{Get-Seq} in Figure~\ref{fig:getseq} works (see also Appendix~\ref{apx:example}):

\myparab{Get-Seq.}  \proc{Get-Seq} deterministically constructs the activation function $T(\cdot)$ from a seedset $S$ and flow network $\mathcal H$ by first activating all seeds $u\in S$ at timestep $t=1$.  \proc{Get-Seq} then iterates over each timestep $t$, and activating every nonseed node $u \notin S$ where (a) time $t$ is after its threshold, \ie $t\geq \theta(u)$, and (b) there is an edge in $\mathcal H$ to $X_{u,t}$ from  some other $X_{v,t'}$ such that node $v$ is active at time $t'<t$.  Observe that the iterative nature of this procedure, along with the structure of $\mathcal H$ and the fact that the seedset $S$ is connected,  implies that there is also a path in $\mathcal H$ from $X_{1,1}$ to $X_{u,t}$ consisting of ``active'' vertices $X_{v,\tau}$, \ie vertices in $\mathcal H$ such that $T(v) \leq \tau$.

We next introduce a few definitions.  First, for each pair $u$ and $t$ (where $t \geq \theta(u)$), let an arbitrary (but fixed) solution $\mathcal F_{u,t}$ for the $(u,t)$-flow problem be the \emph{representative flow} for the $(u,t)$-flow problem. To help us understand how disjoint flows merge, we use the following notion:

\begin{definition}[Border nodes] Consider the $(u, t)$-flow problem on the flow graph $\mathcal H$ and the corresponding representative flow $\mathcal F_{u,t}$. Let us decompose the flow into paths (in an arbitrary but consistent manner) $\mathcal P_1$, $\mathcal P_2$, ..., $\mathcal P_q$. Consider an arbitrary path $\mathcal P_k$ and let $X_{j,\tau}$ be the \emph{last} node on $\mathcal P_k$ that is to the left of the threshold line. Define $X_{j,\tau} = \border(\mathcal P_k)$.
\begin{itemize}
\item The border nodes for the $(u,t)$-flow problem on flow graph $\mathcal H$ are
$\beta(u, t) \triangleq \{\border(\mathcal P_1), ..., \border(\mathcal P_q)\}$.
\item The border nodes for the $(u,t)$-flow problem on $G$ are $B(u,t) \triangleq  \{v_j: \exists \tau \mbox{ s.t. } X_{j,\tau} \in \beta(u,t)\}$.
\end{itemize}
For notational convenience, when $t < \theta(u)$, we let $\beta(u,t) = B(u,t) = \emptyset$.
\end{definition}

\noindent An expository example of $G$, $\mathcal H$ and their border nodes is
in Appendix~\ref{apx:example}.
Border nodes are useful because $\proc{Get-Seq}$ ensures any nonseed node $u$ activates at time $t>\theta(u)$ whenever a border node in $B(u,t)$ is in the seedset $S$.  Letting $p_j$ be the probability that node $v_j$ is placed in the $S$ in a single run of Approach 3, and defining the \emph{seed weight} of node $v_j$ as $\omega_j \triangleq \sum_{t < \theta(v_j)}x_{j,t}$ so that $p_j=\min\{1,\alpha \omega_j\}$ (recall that $\alpha$ is our sampling bias in Approach 3), it follows that $\Pr[T(u)\leq t]$ is related to $\sum_{v_j\in B(u,t)}\omega_j$.  The following lemma therefore allows us to relate $\Pr[T(u)\leq t]$ to the demand in the $(u,t)$-flow problem $\sum_{\theta(u) \leq \tau \leq t}x_{u,\tau}$, which is the main task of the proof of Proposition~\ref{prop:flow}:

\begin{lemma}[Border node lemma]\label{lem:border}
$\sum_{v_j \in B(u,t)}\omega_j \geq \sum_{\theta(u) \leq \tau \leq t}x_{u,\tau}$
for any $u\in V$ and $t\geq\theta(u)$.
\end{lemma}
\noindent
This lemma, proved in Appendix~\ref{apx:borderLemma}, uses the fact the demand of the $(u,t)$-flow problem is upperbounded by the total capacity of the border nodes $B(u,t)$, which is in turn upperbounded by the total seed weight of the border nodes.  Armed with our border node lemma, we can move on to our main task:
\begin{proof}[Proof of Proposition~\ref{prop:flow}]
One can verify, by the construction of $\proc{GetSeq}$, that the activation function $T$ is always connected and for any $u \notin S$, $T(u) \geq \theta(u)$.  Our main objective here is to prove that $\Pr[T(u)\leq t] \propto \sum_{\tau<t}x_{i,t}$ for every pair $(u,t)$ where $u \in V$ and $t \leq n$.  More specifically, we need to show that:
\begin{equation}\label{eqn:toprove}
\begin{array}{l}
 \mbox{\textbf{Part 1.} If $\sum_{t' \leq t}x_{u,t'}\geq \frac{1}{12(1+\epsilon)}$,}
\ifnum\conference=1  \\ \quad \quad \fi
 \mbox{ then } \Pr[T(u) \leq t] \geq 1 - \frac{\epsilon}{4n^2}.\\
 \mbox{\textbf{Part 2.} If $\sum_{t' \leq t}x_{u,t'}< \frac{1}{12(1+\epsilon)}$,}
\ifnum\conference=1  \\ \quad \quad \fi
  \mbox{ then } \Pr[T(u) \leq t] \geq (1+\epsilon)(\sum_{t' \leq t}x_{u,t'}).
 \end{array}
\end{equation}

\noindent Our proof relies on the observation that $T(u) \leq t$ if at least one of the following events hold:
\begin{description}
\item[$\me_1$:] $u$ is seed (because $\proc{GetSeq}$ activates all seeds at $t=1$)
\item[$\me_2$:] $\exists$ an active border node $v_j \in B(u,t)$ in $G$.  ($\me_2$ implies there exist $\tau <t'\leq t$ such that the border node $X_{j,\tau}$ in $\mathcal H$ is active and $\proc{GetSeq}$ will activate node $X_{u,t'}$ and $u$ activates by time $t$.)
\end{description}
\smallskip

\noindent
We now use the relationship between the capacity of the border nodes and the demand of the $(u,t)$ flow problem (namely, $\sum_{\theta(u)<\tau\leq t}x_{u,\tau}$) to prove Part 2 of (\ref{eqn:toprove}). The proof of Part 1 uses similar techniques, and is deferred to Appendix~\ref{sec:missingpart1}.
Given our observation above we have:

\ifnum \conference=0
{\small
\begin{align*}
& \Pr[T(u)\leq t] \\
& \geq \Pr[ \me_1 \vee \me_2] \notag \geq 1 - \min\{\Pr[\neg\me_1], \Pr[\neg \me_2]\} \notag\\
&\geq 1- \min\Big\{\Pr[\neg \me_1], 1-2(1+\epsilon)\big(\sum_{\theta(u) \leq \tau \leq t}x_{u,\tau}\big)\Big\}
& \text{(Lemma~\ref{lem:bordertoi})}\\
&\geq 1 - \min\Big\{1-2(1+\epsilon)\sum_{\tau < \theta(u)}x_{u,\tau},1-2(1+\epsilon)(\sum_{\theta(u) \leq \tau \leq t}x_{u,\tau})\Big\}
&\text{(Since $\alpha \geq 2(1+\epsilon)$)} \\
& \geq  \max\Big\{ 2(1+\epsilon)\sum_{\tau < \theta(u)}x_{u,\tau}, 2(1+\epsilon)(\sum_{\theta(u) \leq \tau \leq t}x_{u,\tau})\Big\} \geq (1+\epsilon)\sum_{\tau \leq t}x_{u,\tau}
\end{align*}
}
\else
{\small
\begin{equation}\label{eqn:Tleqt}
 \Pr[T(u)\leq t]  \geq \Pr[ \me_1 \vee \me_2]  \geq 1 - \min\{\Pr[\neg\me_1], \Pr[\neg \me_2]\}
\end{equation}
In the forthcoming Lemma~\ref{lem:bordertoi}, we shall develop a bound on $\Pr[\neg \me_2]$. Together with the fact that $\alpha \geq 2(1+\epsilon)$, we have
\ifnum\conference=0
\begin{eqnarray*}
& & \min\{\Pr[\neg\me_1], \Pr[\neg \me_2]\} \\
& \leq &  \min\Big\{1-2(1+\epsilon)\sum_{\tau < \theta(u)}x_{u,\tau},1-2(1+\epsilon)(\sum_{\theta(u) \leq \tau \leq t}x_{u,\tau})\Big\}.
\end{eqnarray*}
\else
{\small
\begin{eqnarray*}
& & \min\{\Pr[\neg\me_1], \Pr[\neg \me_2]\} \\
& \leq &  \min\Big\{1-2(1+\epsilon)\sum_{\tau < \theta(u)}x_{u,\tau},\\
& & \quad \quad 1-2(1+\epsilon)(\sum_{\theta(u) \leq \tau \leq t}x_{u,\tau})\Big\}.
\end{eqnarray*}
}
\fi
Therefore, we can continue the analysis for (\ref{eqn:Tleqt}) and get:
\begin{eqnarray*}
& &  \Pr[T(u)\leq t] \\
& \geq &  \max\Big\{ 2(1+\epsilon)\sum_{\tau < \theta(u)}x_{u,\tau}, 2(1+\epsilon)(\sum_{\theta(u) \leq \tau \leq t}x_{u,\tau})\Big\} \\
& \geq & (1+\epsilon)\sum_{\tau \leq t}x_{u,\tau}
\end{eqnarray*}
}
\fi

Lemma~\ref{lem:bordertoi} applies
the border node Lemma~\ref{lem:border} to relate the probability that at least one border node is in the seedset (\ie $\Pr[\neg\me_2]$) with the demand of the $(u,t)$-flow problem (\ie $\sum_{\theta(u)\leq \tau\leq t}x_{u,\tau}$). Specifically:

\begin{lemma}\label{lem:bordertoi}  For every $u \in V$ and $t\in[n]$ where
$\sum_{\tau \leq t}x_{u,\tau}\leq \frac 1{12(1+\epsilon)}$ we have
\vspace{-2mm}$$\Pr[\me_2]=1-\prod_{v_j\in B(u,t)}(1-p_j) \geq 2(1+\epsilon)\sum_{\theta(u)\leq \tau\leq t}x_{u,\tau}\vspace{-2mm}$$
\end{lemma}

\begin{proof}[Sketch of proof]
The idea here is to use a first order approximation of the polynomial $\prod_{v_j \in B(i,t)}(1-p_j)$. For the purpose of exposition, let us assume
at the moment that $\alpha$ is negligibly small and
\vspace{-2mm}
$$\sum_{v_j \in B(u,t)}p_j \approx \sum_{v_j \in B(u,t)}\omega_j \approx \sum_{\theta(u)\leq t' \leq t}x_{u,t'} < \frac{1}{12(1+\epsilon)}
\ifnum \conference=0
\ll 1 \vspace{-2mm}
\fi
$$
where second approximation uses the border node Lemma~\ref{lem:border}. We can then approximate the polynomial $\prod_{v_j \in B(i,t)}(1-p_j)$ by its first order terms, \ie. $\prod_{v_j \in B(u,t)}(1-p_j)\approx 1 - \Theta(\sum_{v_j \in B(u,t)}p_j) \approx 1-\Theta(\sum_{\theta(u)\leq t' \leq t}x_{u,t'} )$, which would complete the proof.

The problem with this argument is that Lemma~\ref{lem:border} only guarantees that $\sum_{v_j \in B(u,t)}p_j > \sum_{\theta(u)\leq t' \leq t}x_{u,t'}$. When $\sum_{v_j \in B(u,t)}p_j$ is substantially larger than the total demand of the $(u,t)$-flow problem, (\eg larger than 1), the first-order approximation becomes inaccurate.  Fortunately, however, we observe that when each individual $p_j$ grows, $\prod_{v_j \in B(u,t)}(1-p_j)$ \emph{decreases}, resulting in even \emph{better} bounds. Roughly speaking, this means when the first-order approximation fails, we are facing an ``easier'' case.   Thus, our strategy will be to \emph{reduce} the case where $\sum_{v_j \in B(u,t)}p_j$ is large, to the case where this quantity is small enough to admit a first-order approximation.
Our implementation of this idea is in Appendix~\ref{apx:bordertoi}.
\ifnum\conference=0
\qed
\fi
\end{proof}
\ifnum\conference=0
\renewcommand{\qedsymbol}{}
\fi
\end{proof}
\ifnum \conference=0
\vspace{-20pt}
\fi

\myparab{Asymptotic optimality of our rounding algorithm. }  We pay a factor of $\ell$ in our rounding algorithm, because we merge $\ell$ different $\{S, T\}$ samples to make sure all the thresholds  are good. But is this really necessary? 
\ifnum \conference=0
In Appendix~\ref{sec:gap2}
\else
In the full paper
\fi
, we show that our rounding algorithm is asymptotically optimal in $\ell$, by presenting an $\Omega(\ell)$ integrality gap for the LP of Figure~\ref{fig:ipflow}. Our problem instance is composed of $\ell$ individual gadgets, where the nodes in gadget $i$ have thresholds chosen from a carefully constructed constant-size set.  We can force these gadgets to be ``independent'', in that sense that if a single $\{S, T\}$ sample causes one of the thresholds in gadget $i$ to be good, we know that whp no threshold in any other gadget can be good.  It follows that merging $\ell$ different $\{S, T\}$ samples, each ensuring that a single threshold is good, is inevitable.

\myparab{Improvement to the approximation ratio.}
Observe that in our rounding procedure we require all $T_j$'s in each of the sampled pairs $\{S_j, T_j\}$ to be feasible (\ie all nodes have to be active at the end of $T_j$). This requirement is not necessary because the merged $T$ will be feasible even if only \emph{one} of the $T_j$ is feasible. We remark here that this observation can be exploited to improve the algorithm so that it returns a feasible seedset of size $\alpha\cdot\opt + \beta$, where $\alpha = O(r(\log n + \ell))$ and $\beta = O(r \ell \log n)$.

\section{Experiments with the IP of Figure~\ref{fig:ip1}}\label{sec:exp}

Given the prevalence of heuristics like ``choose the high degree nodes'' in the literature on technology diffusion in communication networks (\eg \cite{adoptability,JenYannis,adopt}), we sanity-check our approach against several heuristics. Our goal in the following is to \emph{give evidence} that we can find solutions that are substantially different from known heuristics, and to suggest that our IP could be a promising starting point for the design of new heuristics.

We considered problem instances where (a) $G(V,E)$ is 200-node preferential attachment graph with node outdegree randomly chosen from $\{1,2,3,4\}$ \cite{barabassi}, and (b) thresholds $\theta$ randomly chosen from \\$\{\max\{2,c\},2c,3c,..., \lceil \frac{200}{c}\rceil \cdot c\}$.
We ran four groups of experiments with threshold step-length parameter $c$ fixed to $1$, $5$, $10$, and $20$ respectively. For each group, we used a fresh random preferential attachment graph, and repeated the experiment five times with a fresh random instance of the threshold functions.
We solved each of these 20 problem instances using the simple IP formulation presented in Figure~\ref{fig:ip1} (with the extra restriction that the highest degree node must be part of the seedset) and the Gurobi IP solver. We compared the result against  five natural heuristics that iteratively pick a node $u$ with property $X$ from the set of inactive nodes, add $u$ to the seedset $S'$, activate $u$, let $u$ activate as many nodes as possible, and repeats until all nodes are active. We instantiate property $X$ as:\\
(a) \emph{degree:} highest degree,\\
(b) \emph{degree-threshold}: highest (degree)$\times$(threshold),\\
(c) \emph{betweenness:} highest betweenness centrality, \\
(d) \emph{degree discounted:} highest degree in the subgraph induced by the inactive nodes \cite{degDiscount},\\
(e) \emph{degree connected:} highest degree and connected to the active nodes.

\begin{table*}\begin{center}
\hspace{-.5cm}
\ra{1.3}
\begin{footnotesize}
\begin{tabular}{@{}rrrrcrrrcrrr@{}}\toprule
threshold step length: & \multicolumn{2}{c}{ $c= 1$} & & \multicolumn{2}{c}{$c= 5$} &
 & \multicolumn{2}{c}{$c= 10$}& & \multicolumn{2}{c}{$c= 20$}\\
 \cmidrule{2-3} \cmidrule{5-6} \cmidrule{8-9} \cmidrule{11-12}
& Size & Jaccard && Size & Jaccard  && Size & Jaccard && Size & Jaccard\\ \midrule
degree & 11.8 & 0.42  && 20.9 & 0.36  && 24.45 & 0.38 && 41.75 & 0.46 \\
degree-threshold & 8.95 & 0.41  && 15.40 & 0.42  && 19.00 & 0.44 && 33.25 & 0.55 \\
betweenness  & 10.50 & 0.45  && 19.65 & 0.39  && 24.2 & 0.38 && 40.85 & 0.47 \\
degree discounted & 11.2 & 0.39  && 21.55 & 0.34  && 25.35 & 0.36 && 41.60 & 0.45 \\
degree connected & 12.9 & 0.35  && 22.65 & 0.29  && 25.90 & 0.33 && 43.25 & 0.44 \\
ip\_solver & 6.45 & 1  && 11.15 & 1  && 13.75 & 1 && 23.45 & 1 \\
\midrule
degree overlap && 0.44 &  && 0.39 & && 0.37 & && 0.39 \\
betweenness overlap && 0.47 &  && 0.39 & && 0.37 & && 0.40 \\
\bottomrule
\end{tabular}
\end{footnotesize}
\caption{Comparison of the IP of Figure~\ref{fig:ip1} to several heuristics.}\vspace{-3mm}
\label{tab:exp}\end{center}
\end{table*}

For each group, Table~\ref{tab:exp} presents the average seedset size and the average Jaccard index  $\tfrac{|S\cap S'|}{|S \cup S'|}$ between IP seedset $S$ and the heuristic seedset $S'$.  We also compute the fraction of nodes in $S$ that are also part of the top-$|S|$ nodes in terms of (a) degree (the row denoted ``degree overlap''), and (b)  betweenness centrality (``betweenness overlap''). The results of  Table~\ref{tab:exp} do indeed give evidence that our IP can return seedsets that are substantially different (and often better), than the seedsets found via heuristics.

\bigskip\bigskip\bigskip
\section*{Acknowledgements.}
We thank Nadia Heninger, Nicole Immorlica, Prasad Raghavendra, Jennifer Rexford and Santosh Vempala for discussions about earlier incarnations of this model, Michael Mitzenmacher, Michael Schapira and the anonymous reviewers for comments on this draft, and Boaz Barak, Phillipa Gill, David Karger and David Kempe for helpful suggestions.

\newpage
\begin{small}
\bibliographystyle{abbrv}
\bibliography{techDiff}
\end{small}

\appendix

\section{Optimal connected activation sequences provide a 2-approximation}\label{apx:connectSeq}
This section proves Lemma~\ref{lem:connectSequence}. Recall that a connected activation sequence $T$ is such that the set of active nodes at any timestep $t$ induces a connected subgraph of $G$, while a connected seedset is such that all nodes in $S$ induce a connected subgraph of $G$.  Notice that requiring the activation sequence $T$ to be connected is \emph{weaker} than requiring a connected seedset $S$: since $T$ allows a seed to activate \emph{after} a non-seed, the connectivity of $T$ can be preserved by non-seeds whose activation time occurs \emph{between} the activation times of the seed nodes.

We now show that the smallest seedset that gives rise to a feasible connected activation sequence is at most twice the size of the optimal seedset $\opt$.

\begin{proof}[Proof of Lemma~\ref{lem:connectSequence}]
Given an optimal activation sequence $T_\opt$ and seedset $\opt$, we shall transform it into a connected activation sequence $T$. Along the way, we add nodes to the seedset in manner that increases its size by a factor of at most 2.

\myparab{Notation. } Let $G_i(T)$ be the subgraph induced by the first $i$ active nodes in  $T$.  We say a node $u$ is a \emph{connector} in an activation sequence $T$ if the activation of $u$ in $T$ connects two or more disjoint connected components in $G_{T(u) - 1}(T)$ into a single component.

\smallskip
\myparab{Creating a connected activation sequence. } Notice that an activation sequence $T(\cdot)$ is connected if and only if there exists no \emph{connector} in the sequence. Thus, it suffices to iteratively ``remove'' connectors from $T$ until no more connectors remain.

To do this, we initialize our iterative procedure by setting $T\gets T_\opt$. Each step of our procedure then finds the earliest connector $u$ to activate in $T$, adds $u$ to the seedset, and applies the following two transformations (sequentially):

\myparab{  Transformation 1: } First, we transform $T$ so that every component in $G_{T(u)}(T)$ is directly connected to $u$. Let $D(u)$ be the subsequence of $T$ such that every node in
$D(u)$ both activates before $u$, and is part of a component in $G_{T(u)}(T)$ that is \emph{not} connected to $u$.  Transform $T$ so the subsequence $D(u)$ appears immediately \emph{after} node $u$ activates. (This does not harm the feasibility of $T$, because the nodes in $D(u)$ are disconnected from the other nodes in $G_{T(u)}(T)$ that activate before $u$.)

\myparab{  Transformation 2: } Next, we transform the activation sequence so that it is connected up to time $T(u)$. To see how this works, assume that there are only two connected components $C_1$ and $C_2$ in $G_{T(u) - 1}(T)$, where $|C_1| \geq |C_2|$.  Our transformation is as follows:
\begin{enumerate}
\item First, activate the nodes in $C_1$ as in $T(\cdot)$.
\item Then, activate $u$.  (This does not harm feasibility because we added $u$ to the seedset.  Connectivity is ensured because $u$ is directly connected to $C_1$.)
\item Finally, have all the nodes in $C_2$ activate immediately after $u$; the ordering of the activations of the nodes in $C_2$ may be arbitrary as long as it preserves connectivity. (This does not harm feasibility because (a) seed nodes may activate at any time, and (b) any non-seed $v \in C_2$ must have threshold $\theta(v) \leq |C_2| \leq |C_1|$ and our transformation ensures that at least $|C_1|+1$ nodes are active before any node in $C_2$ activates.)
\end{enumerate}
We can easily generalize this transformation to the case where $k$ components are connected by $u$ by letting $|C_1| \geq |C_2| \geq ... \geq |C_k|$ and activate $C_1$, $u$, and the rest of the components sequentially.
At this point, the transformed activation sequence is feasible and connected up to time $t=1+|C_1|+|C_2|+...+ |C_k|$.

\myparab{Seedset growth. } It remains to bound the growth of the seedset due to our iterative procedure. We do this in three steps. First, we observe that the number of extra nodes we added to the seedset is bounded by the number of steps in our iterative procedure.  Next, we iteratively apply the following claim (proved later) to argue that the number of steps in our iterative procedure is upper bounded by number of connectors in the optimal activation sequence, $T_{\opt}$:

\begin{claim}\label{clm:noNewConnectors} Let $T_j$ be the activation sequence at the start of $j^{th}$ step.  The number of connectors in $T_{j+1}$ is less than the number of connectors in $T_{j}$.
\end{claim}
Thus, it suffices to bound the number of connectors in $T_\opt$.  Our third and final step is to show that the number of connectors in $T_\opt$ is bounded by $|\opt|$. To do this, we introduce a potential function $\Phi(t)$ that counts the number of disjoint connected components in $G_{T_\opt(t)}(T)$, and argue the following:
\begin{itemize}
\item For every connector $u$ that activates at time $t$ in $T_\opt$ and joins two or more  components, there is a corresponding decrement in $\Phi$, \ie $\Phi(t)\leq \Phi(t-1)-1$.

\item Next, we have that $\Phi(1)=\Phi(|V|) = 1$, since at the first timestep, there is only one active node, and at the last timestep all the nodes in the graph are active and form a single giant component.  Thus, for every unit decrement in $\Phi$ at some time $t$, there is a corresponding unit increment in $\Phi$ at some other time $t'$.

\item Finally, for any unit increment in $\Phi$, \ie $\Phi(t')=\Phi(t'-1)+1$, it follows that a new connected component appears in $G_{T_\opt(t')}(T)$. This implies that a new seed activates at time $t'$.  Thus, it follows that the number of unit decrements of $\Phi$ is upperbounded by the size of the seedset $|\opt|$.
\end{itemize}
Thus, we may conclude that the number of connectors added to the seedset in our iterative procedure is upperbounded by the number of connectors in $T_\opt$ which is upperbounded by the size of the optimal seedset $\opt$, and the lemma follows.
\end{proof}

The correctness of Claim~\ref{clm:noNewConnectors} is fairly intuitive, given that our transformations always preserve the ordering of the nodes that are not in the components joined by node $u$.  For completeness, we include the proof 
\ifnum\conference=1
in the full version.
\else
here.

\begin{proof}[Proof of Claim~\ref{clm:noNewConnectors}]
We make use of the following observation:

\myparab{Observation 1:} If two activation sequences $T$ and $T'$ have a common suffix, \ie $T = T'$ for timesteps $\tau,\tau+1,...,|V|$, then $T$ and $T'$ contain the same number of connectors after time $\tau-1$.

Let $t=T_j(u)$, where $u$ is the earliest connector in $T_j$. By construction, no connectors exist in $T_j$ prior to time $t$.  Furthermore, we can use Observation 1 to argue that $T_j$ and $T_{j+1}$ contain the same number of connectors after time $t$. Thus, it suffices to show that Transformations 1 and 2 in the $j^{th}$ step of our iterative procedure do not introduce new connectors that activate in prior to time $t$.

Let $T^*$ be the activation sequence after Transformation 1 in the $j^{th}$ step of our iterative procedure, and let $t'=T^*(u)$. We can see that (1) no new connectors activate before time $t'$  in $T^*$ (since, before $t'$ our construction ensures that $T^*$ consists only of active components that are joined by $u$) and (2) no new connectors activate between time $t'+1$ and $t$ inclusive (since (a) $u$ was chosen as the earliest connector in $T_j$, and (b) Transformation 1 preserves the order of the nodes that activate between time $t'+1$ and $t$ inclusive in $T^*$).

Finally, we conclude by arguing that Transformation 2 cannot introduce new connectors by (1) applying Observation 1 to the nodes after $t'$ and (2) observing that after Transformation 2, the nodes that activate before $t'$ create a single connected component, and thus by definition cannot contain any connectors.
\end{proof}

\fi

\section{The augmented integer program (proof of Lemma~\ref{lem:flowip})}\label{apx:flowIP}

We prove Lemma~\ref{lem:flowip} in three parts. First, we show that if we add following two constraints to the IP in Figure~\ref{fig:ip1}:  (a) $x_{1,1} = 1$ and (b) that $(i, t)$-flow problems have feasible solutions for all $i$ and $t \geq \theta(v_i)$, then the resulting IP returns the subset of solutions of the original IP where $T(v_1) = 1$. We also remark on how to remove the assumption that $T(v_1) = 1$ in the optimal $T$.
Second, we show how to encode the flow constraints as an IP.
Finally, we mention why the corresponding relaxed LP is efficiently solvable.

\myparab{Part 1. Any connected activation sequence satisfies the flow constraints}  It suffices to show that for any connected activation sequence, its corresponding integral variables
$\{x_{i,t}\}_{i,t\leq n}$ satisfy the $(i,t)$-flow constraints for all $i$ and $t \geq \theta(v_i)$.

In what follows, we both use $\{x_{i,t}\}_{i,t\leq n}$ and $T(\cdot)$ to represent the activation sequence. Let $\{x_{i,t}\}_{i,t\leq n}$ be a connected activation
sequence. Let us consider an arbitrary $(i, t)$-flow. Let $\tau$ be the time step such that $x_{i,\tau} = 1$.  Recall that the demand in an $(i,t)$-flow problem is $\sum_{\theta(u_i) \leq \tau \leq t}x_{i,\tau}$.
Therefore, when $\tau > t$ or $\tau < \theta(v_i)$, the demand is $0$ and we are done.  We only need to consider the case where $\theta(v_i) \leq \tau \leq t$.  
We claim that when $\{x_{i,t}\}_{i,t\leq n}$ is a connected activation sequence, for any $t$ and $v_k \triangleq T^{-1}(t)$, there exists a path
$v_1v_{i_1}v_{i_2}...v_{i_j}v_k$ such that
$$T(v_1) < T(v_{i_1}) < ... < T(v_{i_{j-1}}) < T(v_{i_{j}}) < T(v_k) = t.$$
This can be seen by induction on $t$. For the base case, $t = 2$ and the path is $v_1v_{k}$. For the induction step, suppose the claim holds for every time step up to $t - 1$. We show that it also holds when $v_{k}$ activates at the $t$-st time step.
Since $\{x_{i,t}\}_{i,t\leq n}$ is connected, there exists a $v_{k'}$ such that there is an edge $\{v_{k'}, v_{k}\} \in E$ and $T(v_{k'}) < T(v_{k})$. By the induction hypothesis, there must be a path $v_1...v_{k'}$ that connects $v_1$ and $v_{k'}$, where the activation time of each node on the path increases monotonically. Thus, the path we seek is $v_1v_{i_1}...v_{k'}v_k$, which completes the proof of the induction step.

We conclude the proof by using the claim we proved by induction.  Namely, there is a path from $v_1$ to  the node $v_{k}$ activating at time $t$.  It follows that we can we push a unit of flow along the path induced in the flow graph $\mathcal H$, namely $X_{1,1}$, $X_{i_1, T(v_{i_1})}$,
 ..., $X_{i_{j}, T(v_{i_j})}$, $X_{k,t}$, so we must have a feasible solution to the $(i,t)$-flow problem.

\mypara{Turning on $v_1$.} We remark that while we have been assuming that $v_1$ is known to activate at $t=1$ in the optimal solution, we can ensure this assumption holds by polynomial-time ``guessing''; run the IP $O(|V|)$ times, relabeling a different node in the graph as $v_1$ in each run, and use the run that returns the smallest seedset.  

\myparab{Part 2. Implementation of the flow constraints.} The $(i,t)$-flow constraints
are enforced via the max-flow-min-cut theorem, \ie by using the fact that the minimum cut between the source and the sinks is the same as the maximum flow. Thus, to ensure every $(i,t)$-flow problem has
a feasible solution, we require the capacity for all the cuts between the source and the sinks to be larger than the demand. The actual implementation is quite straightforward, but we present the details of the IP for completeness:
\begin{itemize}
\item The capacity constraints we have are over the nodes in $\mathcal H$. We use standard techniques to deal with this: we replace each node $X_{i,t}$ in $\mathcal H$ with two nodes $X^+_{i,t}$ and $X^-_{i,t}$ connected by a directed edge of capacity $x_{i,t}$.

\item There are multiple sinks in a $(i, t)$-flow problem. To deal with this, for every $i$ and $t \geq \theta(v_i)$, we introduce a new node $\sink_{i,t}$ to $\mathcal H$ that is connected to every sink $X_{i,\theta(v_i)}$, $X_{i,\theta(v_i) + 1}$, ..., $X_{i,t}$ that sinks \emph{all} the flow in the $(i,t)$-flow problem.
\end{itemize}

\noindent Our implementation is presented in Figure~\ref{fig:ipflow}. Let $S$ and $\overline S$ be two arbitrary partition of the nodes in $\mathcal H$. We let $\delta(S, \overline S)$ be the cut of the partition, \ie the set of edges whose end points are in different subsets of the partition. Also, we let $c(e)$ be the capacity of the edge $e$, \ie $c(\{X^+_{i,t}, X^-_{i,t}\}) = x_{i,t}$ and $c(e) = \infty$ for all other edges.

\ifnum\conference=0
\begin{figure}
\else
\begin{figure*}
\fi
\vspace{3.5cm}
\centering
\begin{pspicture}(0,-3.5)
\psframebox{
$\begin{array}{llll}
 \min  & \sum_{i \leq n}\sum_{t < \theta_i}x_{i,t} &\\
\mbox{subject to:} \\
 \forall i, t: & x_{i, t} \in \{0, 1\} \\
\forall i & \sum_{t \leq n}x_{i, t}  =   1& \mbox{(permut'n constraints)}\\
\forall t &  \sum_{i \leq n}x_{i, t}   =   1& \mbox{(permut'n constraints)}\\
\forall t>1,i: & \sum_{\{v_{i}, v_{i'} \in E\}}\sum_{t' < t}x_{i',t'}  \geq  x_{i,t} & \mbox{(connectivity constraints)}
\\
&  x_{1, 1}   =   1& \mbox{(make $X_{1,1}$ the source)}\\
 {\footnotesize \forall i, t\geq\theta(u_i) \;\forall \mbox{ partitions of $V(\mathcal H)$ } } & &\\
  {\footnotesize  S, \overline{S}, \mbox{s.t.$ X^+_{1,1} \in S$, $\sink_{i,t} \in \overline S$}} & \sum_{e \in \delta(S, \overline S)}c(e) \geq \sum_{\theta(u_i) \leq t' \leq t}x_{i, t'} & \mbox{(flow constraints)}.
\end{array}$}
\end{pspicture}
\vspace{-15mm}
\caption{Integer program for solving the technology diffusion problem.}
\label{fig:ipflow}
\ifnum\conference=1
\vspace{-3mm}
\end{figure*}
\else
\end{figure}
\fi

\myparab{Part 3. The relaxed linear program is efficiently solvable.} Our relaxed LP contains an exponential number of constraints (namely, the flow constraints). Nevertheless, we can use the ellipsoid method to find an optimal solution in polynomial time using a separation oracle~\cite{WS10} that validates if each of the $(i, t)$-flow problems over $\mathcal H$ have solutions, and if not, returns a min-cut constraint that is violated. This oracle can be constructed using algorithms in, \eg \cite{HO92}.

\section{Missing proofs for Section~\ref{sec:roundingAlgo}}\label{apx:missing}

\subsection{Proof of Lemma~\ref{lem:sample} (Success of a single trial)}\label{apx:proofLemSample}

Recall that $A_t$ is number of active nodes by time $t$ (inclusive). We have \Snote{This equation below was buggy before, I think I fixed it}
\ifnum \conference=0
\begin{equation}\label{eqn:nbound}
\Pr[A_n < n] = \Pr[\exists v: T(v) > n] \leq \sum_{v \in V}\Pr[T(v) > n] \leq \frac{n\epsilon}{4n^2} = \frac{\epsilon}{4n}.
\end{equation}
\else
\begin{eqnarray*}
\Pr[A_n < n] &=& \Pr[\exists v: T(v) > n]\\
& \leq & \sum_{v \in V}\Pr[T(v) > n] \\
&\leq & \frac{n\epsilon}{4n^2} \\
&= &\frac{\epsilon}{4n}.
\end{eqnarray*}
\fi
The last inequality holds because of Proposition~\ref{prop:flow}.  It suffices to show that $\Pr[A_t \geq t] \geq \frac{3\epsilon}{4n}$ since $\Pr[A_n = n \wedge A_t \geq t] \geq \Pr[A_t \geq t] - \Pr[A_n < n]$.

Let us partition $V$ into heavy nodes $H$, and light nodes $L$.  We put $v \in H$ when $\sum_{\tau \leq t}x_{v,\tau} \geq \frac{1}{12(1+\epsilon)}$, and $v \in L$ otherwise.  Let's consider two cases, based on the ``weight'' of the light nodes $\rho_t$:
\begin{equation}\label{eq:rhoT}
\rho_t = \sum_{v \in L}\sum_{\tau \leq t}x_{v,\tau}
\end{equation}

\myparab{Case 1. $\rho_t < 1$} \emph{(The light nodes are very light).} Recalling that the permutation constraints of our LP impose that $\sum_{v\in V} \sum_{\tau<t} x_{v,\tau} = t$, it follows that

\ifnum \conference=0
$$t-1 < t - \rho_t = \sum_{v\in V} \sum_{\tau < t} x_{v,\tau} -\sum_{v \in L}\sum_{\tau \leq t}x_{v,\tau} = \sum_{v \in H}\sum_{\tau \leq t}x_{v,\tau}   \leq t$$
\else
\begin{eqnarray*}
& &t-1 \\
&<& t - \rho_t\\
& = &\sum_{v\in V} \sum_{\tau < t} x_{v,\tau} -\sum_{v \in L}\sum_{\tau \leq t}x_{v,\tau} \\
&= &\sum_{v \in H}\sum_{\tau \leq t}x_{v,\tau}\\
& \leq & t.
\end{eqnarray*}
\fi
Using the first and last inequalities and taking the ceiling, we get that
$|H| \geq \left\lceil \sum_{v \in H}\sum_{\tau \leq t}x_{v,\tau} \right\rceil = t.$
Since $|H| \geq t$, if every node in $H$ activates before time $t$ we know that $A_t \geq t$. We write
\ifnum \conference=0
\begin{equation}\label{eqn:Atgt}
\Pr[A_t \geq t] \geq \Pr[ T(v) \leq t, \forall v \in H]\geq 1 - \sum_{v \in H}\Pr[T(v) > t] \geq 1 - \frac{\epsilon}{4n},
\end{equation}
\else
\begin{eqnarray*}
\Pr[A_t \geq t] & \geq & \Pr[ T(v) \leq t, \forall v \in H]\\
&\geq& 1 - \sum_{v \in H}\Pr[T(v) > t] \\
&\geq& 1 - \frac{\epsilon}{4n},
\end{eqnarray*}
\fi
where the last inequality
\ifnum \conference=0
 in (\ref{eqn:Atgt}) 
\fi
holds because of Proposition~\ref{prop:flow}.

\myparab{Case 2. $\rho_t \geq 1$}\emph{(The light nodes are not very light).}  We start by defining two events.
\begin{description}
\item{$\me_1$} is the event that all the heavy nodes are active by time $t$, \ie $T(v) \leq t \; \forall v \in H$.
\item{$\me_2$} is the event that at least $\rho_t$ light nodes are on by time $t$, \ie $|\{ v \in H \wedge T(v) \leq t \}| > \rho_t$. 
\end{description}
When both $\me_1$ and $\me_2$ occur, we have
$$A_t \geq |H| + \rho_t \geq \sum_{v\in H} \sum_{\tau \leq t} x_{v,\tau} +\sum_{v\in L} \sum_{\tau \leq t} x_{v,\tau} = t$$
where both the second inequality and the last equality use the permutation constraints of the LP.
\Snote{I just added this following sentence in, Zhenming can you verify it?}
It follows that $\Pr[A_t > t] \geq \Pr[\me_1 \wedge \me_2] \geq \Pr[\me_2] - \Pr[\neg \me_1]$.  We now bound each event individually.

Let's start by bounding $\Pr[\me_2]$.
Letting $I(\cdot)$ be an indicator variable that sets to 1 iff the parameter is true, we have that
\begin{equation}\label{eq:babab}
\begin{array}{lll}
\E[\sum_{v \in L}I(T(v) \leq t)] &= &\sum_{v \in L}\Pr[T(v) \leq t]\\
&\geq& \sum_{v \in L}\left((1+\epsilon)\sum_{t'\leq t}x_{v,t'}\right) \\
&= &(1+\epsilon)\rho_t
\end{array}
\end{equation}
where the inequality uses Proposition~\ref{prop:flow} as usual. Meanwhile, using the law of total probability we get
\begin{equation}\label{eq:rarar}
\E[\sum_{v \in L}I(T(v) \leq t)] \leq \Pr[\me_2]n + \Pr[\neg \me_2]\rho_t \leq \Pr[\me_2]n + \rho_t
\end{equation}
Combining (\ref{eq:babab})-(\ref{eq:rarar}) we find that $\Pr[\me_2] \geq \frac{\epsilon \rho_t} n \geq \frac{\epsilon}{n}$.
Next, we bound $\Pr[\me_1]$ by observing that
$$\Pr[\neg \me_1] \leq \sum_{v \in H}\Pr[T(v) > t] \leq \tfrac{\epsilon}{4n}$$
using Proposition~\ref{prop:flow} for the last inequality again.
Finally, we combine both bounds to conclude that $\Pr[A_t > t] \geq \Pr[\me_1 \wedge \me_2] \geq \Pr[\me_2] - \Pr[\neg \me_1] \geq \frac{3\epsilon}{4n}$ as required.

\subsection{Proof of Lemma~\ref{lem:border} (Border node lemma)}\label{apx:borderLemma}
Let us decompose the representative flow $\mathcal F_{i,t}$ into paths (in an arbitrary but consistent manner) $\mathcal P_1$, $\mathcal P_2$, ..., $\mathcal P_q$, and let $f_k$ be the volume of the flow on path $P_k$.
\ifnum \conference=0
\begin{align*}
\sum_{\theta(v_i)\leq \tau \leq t} x_{i,\tau}
& = \sum_{k } f_k & \text{ (the demand in the $(i,t)$ flow problem is satisfied)}\\
& = \sum_{ X_{j,\tau}  \in \beta(i,t)} \; \sum_{\border(P_k) = X_{j,\tau} } f_k &\text{(multiple $\border(P_k)$ can map to a single $X_{j,\tau}$)}\\
& \leq \sum_{ X_{j,\tau}  \in \beta(i,t)} x_{j,\tau} &\text{(bounding capacity of $X_{j,\tau}$)}\\
& = \sum_{ v_j \in B(i,t)}\;\sum_{\tau \text{ s.t. } X_{j,\tau} \in \beta(i,j)} x_{j,\tau} &\text{(translating from $\mathcal H$ to $G$)}\\
& \leq \sum_{ v_j \in B(i,t)}\;\sum_{\tau \leq \theta(v_j)} x_{j,\tau}
&\left(\tau \text{ s.t. } X_{j,\tau} \in \beta(i,j) \Rightarrow \tau \leq \theta(v_j)\right)\\
& = \sum_{ v_j \in B(i,t)} w_j
&\text{(definition of $w_j$)}
\end{align*}
\else
\begin{align*}
\sum_{\theta(v_i)\leq \tau \leq t} x_{i,\tau}
& = \sum_{k } f_k \\
&\text{\footnotesize (the demand in the $(i,t)$ flow problem is satisfied)}\\
& = \sum_{ X_{j,\tau}  \in \beta(i,t)} \; \sum_{\border(P_k) = X_{j,\tau} } f_k \\
&\text{\footnotesize(multiple $\border(P_k)$ can map to a single $X_{j,\tau}$)}\\
& \leq \sum_{ X_{j,\tau}  \in \beta(i,t)} x_{j,\tau} \\
&\text{\footnotesize(bounding capacity of $X_{j,\tau}$)}\\
& = \sum_{ v_j \in B(i,t)}\;\sum_{\tau \text{ s.t. } X_{j,\tau} \in \beta(i,j)} x_{j,\tau}\\
&\text{\footnotesize(translating from $\mathcal H$ to $G$)}\\
& \leq \sum_{ v_j \in B(i,t)}\;\sum_{\tau \leq \theta(v_j)} x_{j,\tau}\\
&{\footnotesize\left(\tau \text{ s.t. } X_{j,\tau} \in \beta(i,j) \Rightarrow \tau \leq \theta(v_j)\right)}\\
& = \sum_{ v_j \in B(i,t)} w_j\\
&\text{\footnotesize(definition of $w_j$)}
\end{align*}
\fi
Notice that the last four lines give the total seed weight of the border nodes as an upper bound on their total capacity.

\subsection{Proof of Lemma~\ref{lem:bordertoi}}\label{apx:bordertoi}

\begin{proof}[Proof of Lemma~\ref{lem:bordertoi}]
We shall find a non-negative sequence $p'_j$ ($v_j \in B(u,t)$) such that
\begin{itemize}
\item Condition 1: $\prod_{v_j \in B(u,t)}(1-p_j) \leq \prod_{v_j \in B(u,t)}(1-p'_j)$
\item Condition 2: $\sum_{j \in B(u,t)}p'_j = 4(1+\epsilon)\sum_{\theta(u) \leq \tau \leq t}x_{u,\tau}$.
\end{itemize}
When both conditions hold, we can bound $\prod_{v_j \in B(u,t)}(1-p_j)$ by $\prod_{v_j \in B(u,t)}(1-p'_j)$, which
can then be approximated by its first-order terms. We use existential arguments to find the sequences $p'_j$ (for each $v_j \in B(u,t)$):
We start by recalling that $p_j=\min\{1,\alpha\omega_j\}$ and $\alpha > 4(1+\epsilon)$. It follows that when $\omega_j \geq \frac{1}{4(1+\epsilon)}$ for some $v_j \in B(u,t)$, the $p_j = 1$ and the lemma trivially holds. Thus, we may assume that $4(1+\epsilon)\omega_j \leq 1$ for all $v_j \in B(u,t)$, and we can write
{\small
$$\sum_{v_j \in B(u,t)}p_j \geq 4(1+\epsilon)\sum_{v_j \in B(u,t)}\omega_j \geq 4(1+\epsilon)\sum_{\theta(u) \leq \tau\leq t}x_{u,\tau}.$$}
where the second inequality uses Lemma~\ref{lem:border}.

We now know that there exists a sequence $p'_j$ such that $p_j \geq p'_j$ and $\sum_{j \in B(u,t)}p'_j = 4(1+\epsilon)\sum_{\theta(u) \leq \tau \leq t}x_{u,\tau}$, which meets Condition 1 and Condition 2.
It follows that
$\prod_{v_j \in B(u,t)}(1-p_j) \leq  \prod_{v_j \in B(u,t)}(1-p'_j)$, and
we may complete the proof with the following first-order approximation:

\begin{lemma}[First order approximation]\label{lem:approx}Let $x_1, x_2, ..., x_k$ be real positive values such that 
$\sum_{i \leq k}x_i \leq 1$. Then
$$\prod_{i \leq k}(1-x_i) \leq 1 - \frac 1 2\left(\sum_{i \leq k}x_i\right).$$
\end{lemma}
\noindent
When we substitute the $x_i$'s in Lemma~\ref{lem:approx} with $p'_j$s, and use the fact that
$$\sum_{j \in B(u,t)}p'_j = 4(1+\epsilon)\sum_{\theta(u) \leq \tau \leq t}x_{u,\tau} \leq 4(1+\epsilon) \tfrac1{12(1+\epsilon)}=\tfrac13 < 1.$$
we complete the proof because
$\prod_{v_j \in B(u,t)}(1-p'_j) \leq 1 - \frac 1 2 \cdot 4(1+\epsilon)\sum_{\theta(v_j) \leq \tau \leq t}x_{u,\tau}$
\end{proof}

\begin{proof}[Proof of Lemma~\ref{lem:approx} (First order approximation)]
Let $x_1, x_2, ..., x_k$ be real positive values such that \\$\sum_{i \leq k}x_i \leq 1$. Notice that
for any $0 \leq x \leq 1$, we have $(1-x) \leq e^{-x}$. Let $s \triangleq \sum_{i \leq k }x_i$. We have
\ifnum \conference=0
$$\prod_{i \leq k}(1 - x_{i}) \leq \prod_{i \leq k}\exp(x_i) =  \exp(\sum_{i \leq k}x_k) = \exp(s)\leq 1 - s + \frac{s^2}2 \leq 1 - s(1-\frac 1 2) = 1 - \frac s 2.$$
\else
\begin{eqnarray*}
\prod_{i \leq k}(1 - x_{i}) & \leq & \prod_{i \leq k}\exp(x_i) \\
&=&  \exp(\sum_{i \leq k}x_k) \\
&= &\exp(s)\leq 1 - s + \frac{s^2}2 \\
&\leq& 1 - s(1-\frac 1 2) \\
&= &1 - \frac s 2.
\end{eqnarray*}
\fi
\end{proof}

\subsection{First part of Proposition~\ref{prop:flow}}\label{sec:missingpart1}
We now prove the first item in (\ref{eqn:toprove}), \ie we consider a pair $(u,t)$ such that $\sum_{\tau \leq t}x_{u,\tau} \geq \frac{1}{12(1+\epsilon)}$. Let us consider two cases.

\noindent{\emph{Case 1. $\sum_{\tau \leq \min\{\theta(u) - 1, t\}}x_{u,\tau} \geq \frac{1}{24(1+\epsilon)}$}}. In this case, $p_u = 1$ and $u$ is always selected as a seed. Thus, $\Pr[T(u) \leq t] = 1$.

\noindent{\emph{Case 2. $\sum_{\tau \leq \min\{\theta(u) - 1, t\}}x_{u,\tau} < \frac{1}{24(1+\epsilon)}$}} In this case, we can see that $\sum_{\theta(u) \leq \tau \leq t}x_{u,\tau} \geq\frac{1}{24(1+\epsilon)}$.
Therefore, we use the border node Lemma~\ref{lem:border} to get
\begin{equation}\label{eqn:btoi}
\sum_{v_j \in B(u,t)}\omega_j \geq\frac{1}{24(1+\epsilon)}.
\end{equation}
Now, recall that $\Pr[T(u)\leq t] \geq \Pr[\me_2] = 1- \prod_{v_j \in B(u,t)}(1-p_j)$.  Therefore, it suffices to prove that $\prod_{v_j \in B(u,t)}(1-p_j) \leq \frac{\epsilon}{4n^2}$. \Snote{THis last equation used to be geq, which seems like a bug so I fixed it.}

At this point, our analysis deviates from the analysis for the second part of Proposition~\ref{prop:flow}.  There, the $p_j$ values were small enough to allow $\prod_{v_j \in B(u,t)}(1-p_j)$ to be approximated using only first-order terms. Here, we are dealing with the case where $p_j$'s are large. Thus, $\prod_{v_j \in B(u,t)}(1-p_j)$ decays exponentially, and it is more appropriate to approximate it using exponential functions.
By using (\ref{eqn:btoi}) and the following approximation Lemma~\ref{lem:prod} (with $\lambda$ as $\alpha$) we can see that indeed $\prod_{v_j \in B(i,t)}(1-p_j) \leq \frac{\epsilon}{4n^2}$, which completes the proof. \Snote{THis last equation used to be geq, which seems like a bug so I fixed it.}

\begin{lemma}\label{lem:prod}Let $\epsilon$ be an arbitrary
constant. Let $x_1, ..., x_k$ be numbers between $[0, 1]$ such that
$\sum_{i \leq k}x_i= s$, where $s \geq \frac 1 {24(1+\epsilon)}$. Let  $\lambda = 24(1+\epsilon)\ln (\frac{4n^2}{\epsilon})$ and
 $p_i = \min\{\lambda x_i, 1\}$. It follows that
$$\prod_{i \leq n}(1 - \min\{\lambda x_i, 1\}) \leq \frac {\epsilon}{4n^2}.$$
\end{lemma}

\begin{proof}[Proof of Lemma~\ref{lem:prod}]
Let us consider two cases over the values
of $x_i$. In the first case, there exists some $x_i$ such that $\lambda x_i \geq 1$. For this case,
we have
$\prod_{i \leq n}(1 - \min\{\lambda x_i, 1\}) = 0 \leq \frac {\epsilon} {4n^2}.$

In the second case, where all $x_i$ are less than $1/\lambda$, the quantity
$\prod_{i \leq k}(1 - p_i) = \prod_{i \leq k}(1 - \lambda x_i)$ is maximized when $x_1 = x_2 = ... = x_k = \frac s k$.
In other words,
\ifnum \conference=0
$$
\prod_{i \leq k}(1 - \lambda x_i)
 \leq  \left(1 - \frac{\lambda s} k\right)^k\\
  = \left(1 - \frac{\lambda s}{k}\right)^{\frac{k}{\lambda s}\lambda s}\\
  \leq e^{-\lambda s} \\
 \leq  \exp(-\frac{\lambda}{24(1+\epsilon)})  =  \exp(-\ln(\frac{4n^2}{\epsilon})) =  \frac{\epsilon}{4n^2}.
$$
\else
\begin{eqnarray*}
\prod_{i \leq k}(1 - \lambda x_i)
& \leq & \left(1 - \frac{\lambda s} k\right)^k\\
 & = & \left(1 - \frac{\lambda s}{k}\right)^{\frac{k}{\lambda s}\lambda s}\\
 & \leq & e^{-\lambda s} \\
& \leq &  \exp(-\frac{\lambda}{24(1+\epsilon)})  =  \exp(-\ln(\frac{4n^2}{\epsilon})) \\
&= & \frac{\epsilon}{4n^2}.
\end{eqnarray*}
\fi
\end{proof}

\section{Proof of Corollary~\ref{cor:little}}\label{apx:little}
Let us consider an arbitrary technology diffusion problem $\Pi = \{G, \theta\}$.
Let $\Pi^+$ and $\Pi^-$ be the corresponding diffusion problems defined in
Definition~\ref{def:condition}. Recall that $\opt$ is the optimal solution
for $\Pi$, $\opt^+$ is the optimal solution for $\Pi^+$ and $\opt^-$
is the optimal solution for $\Pi^-$. Let
$\mathbf P \triangleq \{\lfloor 1 + \epsilon \rfloor, \lfloor (1 + \epsilon)^2 \rfloor
, ..., \lfloor (1 + \epsilon)^q \rfloor\}$, where $q = \log_{1+\epsilon}n+1$. We next define
a new technology diffusion instance $\Pi' = \{G, \theta'\}$
that uses the same graph and $\theta'(u)$ (for each $u$)
is the smallest number in $\mathbf P$ that is larger than $\theta(u)$. Notice
that $\theta'(u) \leq (1+\epsilon)\theta(u)$. Let $\opt'$ be the size of optimal seed set
for $\Pi'$. We can run our approximation algorithm on $\Pi'$ and get a solution, whose size is at
most $O(\log^2n \cdot r \cdot \opt')$ since the number of thresholds in $\Pi'$ is $\log n$.
Because $\theta'(u) \geq \theta(u)$ for all $u$, a feasible solution in $\Pi'$ is also
a feasible solution in $\Pi$. Thus, the seedset returned by our algorithm is feasible and
$\opt \leq \opt'$. Similarly, we can see that $\opt^- \leq \opt$ and $\opt' \leq \opt^+$.
Therefore, the seedset size can be expressed as $O(\log^2 n \cdot r\opt \frac{\opt'}{\opt})
= O(\log^2 n \cdot r\opt \frac{\opt^+}{\opt^-}) = O(\kappa(\Pi, \epsilon)(\log^2 n) \cdot r \opt)$.

\section{Lower bounds}\label{apx:hard}
\noindent This appendix presents our lower bounds, that can be summarized as follows.
\begin{enumerate}
\item \emph{Computational barrier:} the technology diffusion problem is at least as hard as a Set Cover problem, so that our problem does not admit any $o(ln |V |)$-approximation algorithm.
\item \emph{Combinatorial barrier:} in the worst case, the optimal solution with \emph{a connected seedset} could be $\Omega(r)$ times larger than the optimal solution.
\item \emph{Integrality gaps:} The simple IP (Figure~\ref{fig:ip1} discussed in Section~\ref{sec:linearize}) has an $\Omega(n)$ integrality gap.  The augmented IP (Section~\ref{sec:flowIP} and Figure~\ref{fig:ipflow} of Appendix~\ref{apx:flowIP}) has an $\Omega(\ell)$ integrality gap.
\end{enumerate}

\subsection{Computational barrier}
\Snote{This is the simpler version of this proof, with the join node removed.}

This section proves that the technology diffusion problem is
at least as hard as the set cover problem. Let us recall the definition (of the optimization version) of the set cover problem: given a finite universe $\mathcal U$ and a family $\mathbf S$ of subsets of $\mathcal U$, we are interested in finding the smallest subset $\mathbf T$ of $\mathbf S$ such that $\mathbf T$ is a cover of $\mathcal U$, i.e. $\bigcup_{T \in \mathbf T}T = \mathcal U$. The set cover
cannot be approximated within a factor of $(1-o(1))\ln n$ unless NP has quasi-polynomial time algorithm (see \cite{setcoverApx} and references therein). We have the following lemma.

\begin{lemma}\label{thm:reduction}
Given an $\alpha$-approximation algorithm for the technology diffusion problem with constant number of threshold values $\theta\geq 2$, and constant graph diameter $r\geq 3$, we can obtain an $O(\alpha)$-approximation algorithm for the set cover problem.  Moreover, the reduction holds even if the seedset in the technology diffusion problem is required to be connected.
\end{lemma}

Thus, we can see
that there is no $c\ln n$ approximation algorithm (for some constant $c$) for the technology diffusion problem.
%

\begin{proof}[Proof of Lemma~\ref{thm:reduction}] Let us consider an arbitrary set cover instance $(\mathcal U, \mathbf T)$, where $m = |\mathbf T|$ is the number of sets in $\mathbf T$.

\begin{figure}\begin{center}
\includegraphics[width=4in]{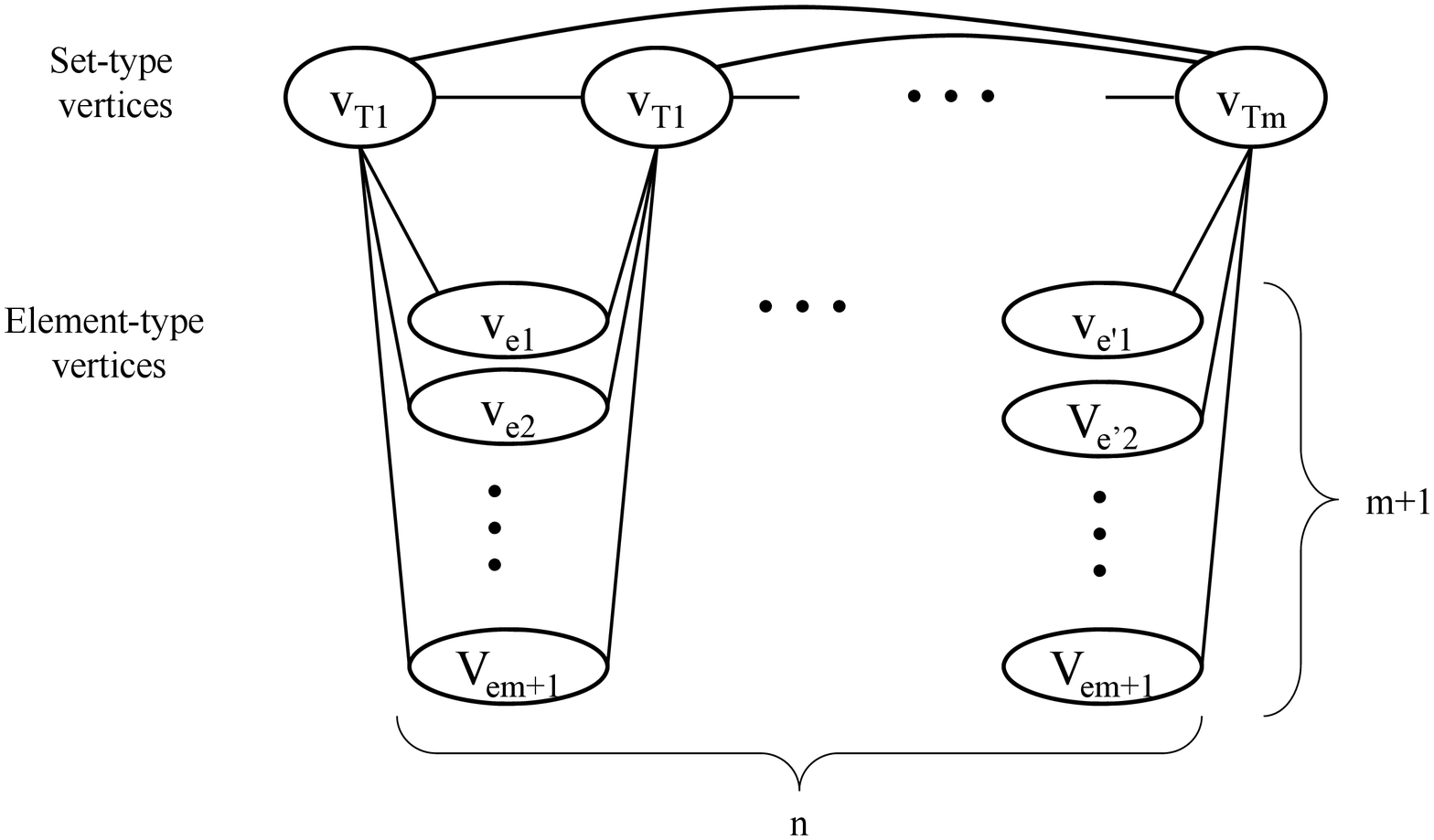}
\caption{Reduction.}\label{fig:reduction}
\end{center}
\end{figure}

\myparab{The reduction. } We construct a technology diffusion problem as described below, and illustrated in Figure~\ref{fig:reduction}:
\begin{itemize}
\item The vertex set consists of the following types of vertices:
\begin{enumerate}
\item The \emph{set type}: for each $T \in \mathbf T$, we shall construct a node $v_T$ in the technology network.
\item The \emph{element type}: for each $e \in \mathcal U$, we shall construct $m + 1$ nodes $v_{e,1}$, $v_{e,2}$, ..., $v_{e, m + 1}$.
\end{enumerate}
\item The edge set consists of the following edges:
\begin{enumerate}
\item For each $T \in \mathbf T$ and $e \in T$, we add the edges $\{v_T, v_{e, 1}\}$, $\{v_T, v_{e, 2}\}$, ..., $\{v_T, v_{e, m + 1}\}$.
\item The set type vertices are connected as a clique. (For each $T\neq T' \in \mathbf T$, we add the edge $\{u_{T}, u_{T'}\}$).
\end{enumerate}
\item The thresholds $\theta(\cdot)$ are set as follows,
\begin{enumerate}
\item For any $e \in \mathcal U$ and $i \leq m + 1$, we set $\theta(v_{e, i}) = 2$.
\item For every $T \in \mathbf T$, we set $\theta(v_{T}) = (m + 1)n + 1$.
\end{enumerate}
\end{itemize}

\myparab{Properties of the reduction. } Notice that our technology diffusion problem has only two types of threshold values. Furthermore, the diameter of the graph we form is exactly $3$ hops (in terms of edges); the maximum distance in this graph is from one $v_{e,i}$ node to another. 
Finally, we show below that the seedset must consist of set-type vertices. Since these vertices form a clique, it follows that there exists
an optimal seedset that is connected.

\myparab{Correctness. } To conclude that the size of the optimal seed set is the same as the size of the optimal cover (which also means that our reduction is approximation-preserving), we establish the following:
\begin{description}
\item[Item 1.] For any feasible cover $\mathbf S$ in the set cover problem, the corresponding seed set $\{v_{S}: S \in \mathbf S\}$ is a feasible solution for the technology diffusion problem.
\item[Item 2.]  Any feasible seedset in the technology diffusion problem \emph{that only consists of set-type vertices} corresponds to a feasible cover in the set cover problem.
\item[Item 3.] Given a feasible seedset that consists of \emph{element type} vertices, 
    there is a feasible seedset of equal or smaller size that consists only of \emph{set type} vertices.  Since the set type vertices form a clique, we have that the optimal solution for the technology diffusion problem is also a connected one.
\end{description}

\myparab{Item 1. }
To show the first item, we simply walk through the activation process: When $\mathcal S$ is a cover, let the seedset be $v_{T_i}$ for all $T_i \in \mathcal S$. Notice that this seedset is \emph{connected}. Upon activating the seedset, the vertices $u_{e, i}$ for all $e \in \mathcal U$ and $i \leq m + 1$ are activated  because they are connected to at least one active seed. Now, there are $(m+1)n$ active nodes, so the rest of the \emph{set type} vertices are activated.

\myparab{Item 2. }
To show the second item, we consider an arbitrary seedset that only consists of the \emph{set type} vertices: $U = \{v_{T_1}, v_{T_2}, ..., v_{T_k}\}$, where $T_1, ..., T_k \in \mathbf T$. We shall show that if $T_1, ..., T_k$ is not a cover, then the seed set cannot be feasible (\ie some nodes will remain inactive in the technology diffusion problem).

Let $e \in \mathcal U/\left(\cup_{j \leq k}T_j\right)$ be an element that is not covered by the sets in $\{T_1, ..., T_k\}$. Let us consider the nodes $v_{e, 1}, v_{e, 2}, ..., v_{e, m + 1}$, and node $v_{T}$ for each $T \notin \{T_1, ..., T_k\}$ in the technology diffusion problem. We claim that none of these vertices will be activated with seedset $U$.  Suppose, for the sake of contradiction, that one or more of these vertices are activated, and consider the first activated vertex among them. There are two cases:
\begin{description}
\item[Case 1.] $v_{T}$ ($T \notin \mathbf T$) is activated first. This is impossible: when $v_{e, i}$ ($i \leq m + 1$) are not activated, the number of activated nodes is at most $(n - 1)(m + 1) + m < (m + 1)n$.

\item[Case 2.] $v_{e, i}$ ($i \leq m + 1$) is activated first. This is impossible because $v_{e, i}$ is only connected with $v_{T}$, where $T \notin \{T_1, ..., T_k\}$ and none these \emph{set type} vertices are activated.
\end{description}

\myparab{Item 3. }
Finally, we move onto the third item. Let us consider a feasible seedset $F$ that does not consist of only \emph{set type} vertices. We show that we can easily remove the \emph{element type} vertices in $F$: let $v_{e, i}$ be an arbitrary vertex in $F$. Then we can remove $v_{e, i}$ from $F$ and add an $v_T$ to $F$ such that $e \in T$. This does not increase the cardinality of $F$. Furthermore, $v_{e, i}$ would still be activated, which implies that the updated $F$ is still be a feasible seed set.
\end{proof}

\subsection{Combinatorial barrier}\label{sec:diameter}
\Znote{Next iteration: change $v_i$ into $u_i$.}

\begin{lemma}\label{lem:limit}For any fixed integer $r$, there exists an instance of technology diffusion problem $\{G, \theta\}$ such that
(a) the diameter of $G$ is $\Theta(r)$, and
(b) the optimal \emph{connected} seedset is at least $\Omega(r)$ larger than the optimal seedset.
\end{lemma}

\begin{proof}[Proof of Lemma~\ref{lem:limit}] Let $r > 0$ be an arbitrary integer. We define graph $G_r$ as follows (see Figure~\ref{fig:limit}):
\begin{itemize}
\item The vertex set is $\{v_1, ..., v_{2r+1}\}$.
\item The edge set is $\left\{\{v_i, v_{i +1}\}: 1 \leq i  < 2r + 1\right\}$.
\end{itemize}
The threshold function shall be defined as follows,
\begin{itemize}
\item $\theta(v_1) = \theta(v_{2r+1}) = 2$ and $\theta(v_{r+1}) = 2r+1$.
\item For $1 < i \leq r$, $\theta(v_i) = i$. \Snote{Zhenming I think there was a bug here, it used to say  $\theta(v_i) = 1$}
\item For $r + 2\leq i < 2r+1$, $\theta(v_i) = 2r + 2 - i$.
\end{itemize}
It is straightforward to see that the diameter of the graph is $2r = \Theta(r)$. It remains to verify
that the optimal connected solution is $\Theta(r)$ times larger than the optimal solution.

It's easy to see that $\{v_1, v_{2r+1}\}$ is a feasible seedset and therefore, the size of the optimal seed set is $O(1)$. We next show that any feasible connected set has size $\Omega(r)$.

Since the seedset must be connected, wlog we can assume that the seedset is $\{v_i, v_{i+1}, ..., v_j\}$ and by symmetry $i \leq r + 1$.
When $j < r + 1$, node $v_{r+1}$ will never activate (because $v_{r+1}$ has threshold $2r+1$, it only activates when all other nodes are active, but in this case all $r$ nodes to the right of $v_{r+1}$ are inactive).  It follows that a feasible seedset requires $j \geq r + 1$.

When $i = 1$, the size of the seedset is $\Theta(r)$ and the lemma follows. So, we need only consider the case where $i > 1$: symmetry allows us to assume wlog that $r+1-i \geq j - (r+1)$ \ie $\theta(v_{j + 1}) \geq \theta(v_{i - 1})$. Therefore, since we have $j-i+1$ nodes in the seedset, a necessary condition for this seedset to be feasible is thus $j - i + 1 \geq i - 2$. Using the fact that $j \geq r + 1$, we get $i \leq r/2 + 2$ and $j - i = \Omega(r)$, which completes our proof.
\end{proof}

One drawback of this construction is that $\ell = \Theta(n)$. We may modify $\theta(\cdot)$ so that $\ell = \tilde O(1)$ (thus ensuring that our lower bound depends on graph diameter $r$, rather than the number of thresholds $\ell$):
\begin{itemize}
\item When $i \leq n$, set $\theta(u_i) = \max\{2^{\lfloor \log_2 i\rfloor}, 2\}$,
\item when $i = n + 1$, set $\theta(u_i) = 2n+1$, and
\item When $i > n$, set $\theta(u_i) = \max\{2^{\lfloor \log_2 (2n+2-i)\rfloor}, 2\}$.
\end{itemize}
One can use similar arguments to show that the size of the optimal seedset is $O(1)$ while the
size of the optimal connected seedset is $\Theta(r)$.

\begin{figure}\begin{center}
\includegraphics[width=4in]{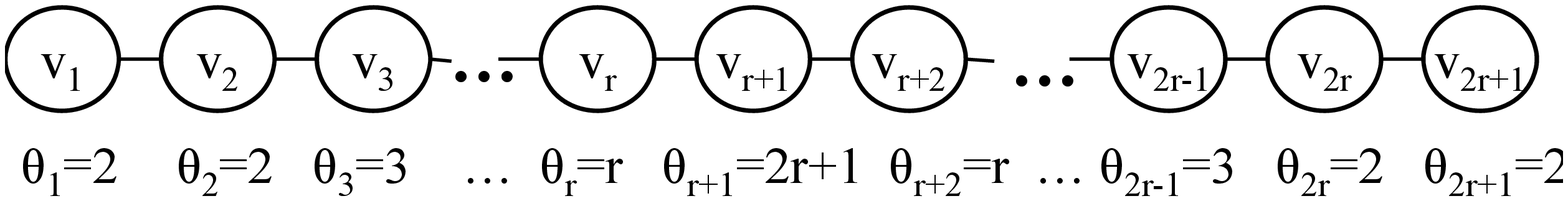}
\caption{An instance of the technology diffusion problem for the proof of Lemma~\ref{lem:limit}.}\label{fig:limit}
\end{center}
\end{figure}

\subsection{Integrality gap}

\subsubsection{Integrality gap for the simple IP of Figure~\ref{fig:ip1}}\label{sec:gap1}
We construct a problem instance with $\ell=O(1)$ where the solution returned by the simple IP of Figure~\ref{fig:ip1} is $O(1)$, while the optimal seedset has size $\Theta(n)$, implying an integrality gap that is polynomial in $n=|V|$.

\myparab{The problem instance.} We let $w$ and $h$ be parameters of the problem instance $\{G,\theta\}$. These parameters control the shape of the graph $G$ and the size of the integrality gap. We will decide the parameters at the end to maximize the integrality gap.  The graph $G$ (see Figure~\ref{fig:ipgap1}) has a node set of size $n=wh+h+1$ that consists of the following nodes:
\begin{itemize}
\item The root node $R$.
\item The ``seed candidates'' $\{s_1, ..., s_h\}$.
\item The ``tail nodes'' $v_{i,j}$ for $i \leq h$ and $j \leq w$.
\end{itemize}

\noindent
The edge set consists of the following two types of edges:
\begin{itemize}
\item all the ``seed candidates'' $s_i$ ($i \in [h]$) are connected with the root $R$.
\item for any specific $i \in [h]$, the nodes $s_i$, $v_{i, 1}$, $v_{i, 2}$, ..., $v_{i, w}$
form a chain. In other words, $\{s_i, v_{i, 1}\} \in E$ and $\{v_{i,j}, v_{i,j+1}\} \in E$ for $1 \leq j \leq w - 1$.
\end{itemize}

\noindent
Hereafter, we shall refer to the chain $s_i$, $v_{i, 1}$, ..., $v_{i, w}$ as the $i$-th tail of the graph.  The threshold function $\theta$ is specified as follows:
\begin{itemize}
\item $\theta(R) = n$.
\item for any $s_i$ we have $\theta(s_i) = n - h + 2$.
\item for any $v_{i,j}$ we have $\theta(v_{i,j}) = 2$.
\end{itemize}

\begin{figure}
\begin{center}
\includegraphics[width=4.4in]{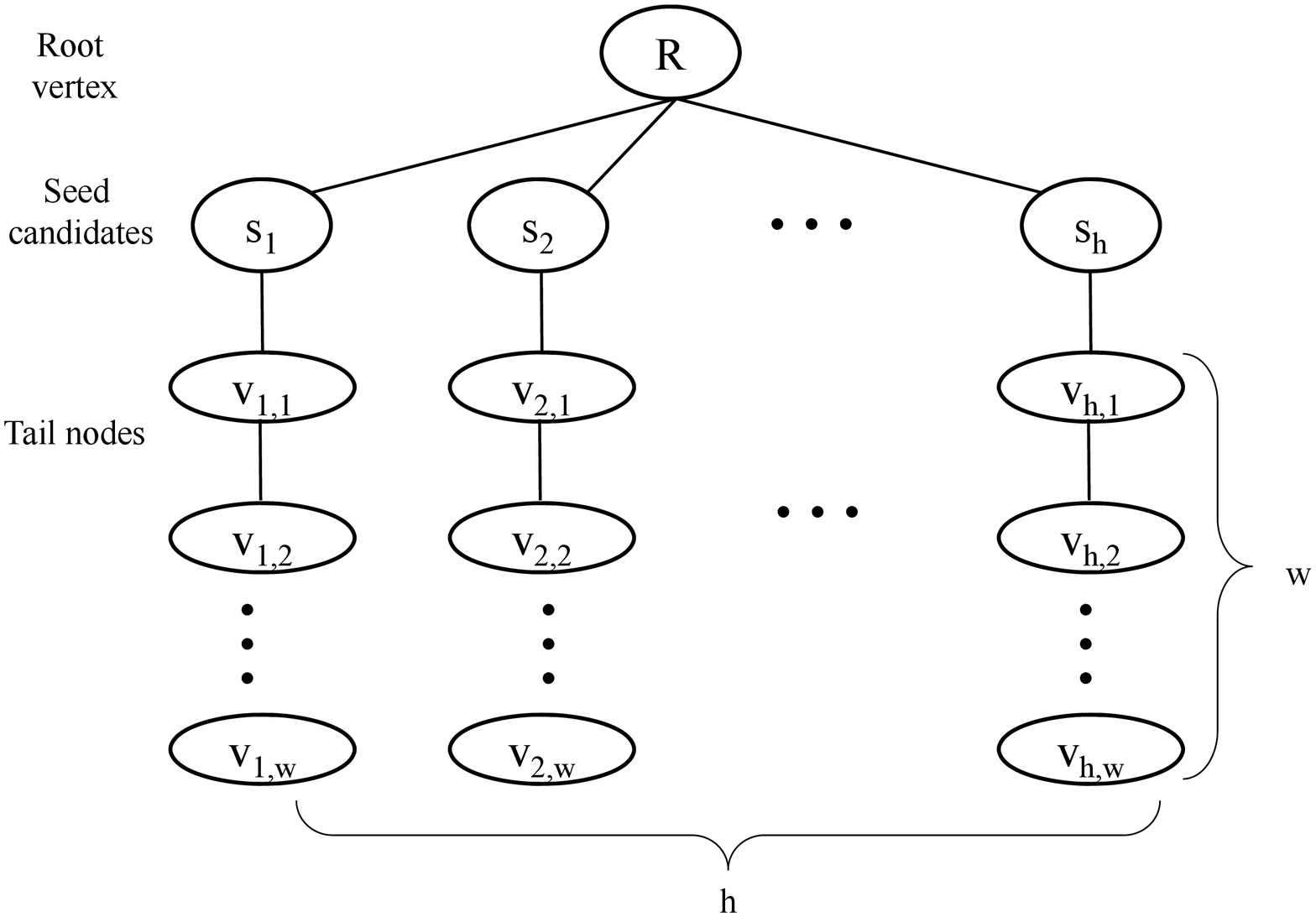}
\caption{The graph for the hard instance with $\Omega(h)$ integrality gap}
\label{fig:ipgap1}
\end{center}
\end{figure}

\smallskip
To exhibit the integrality gap, we shall first construct a feasible fractional solution of constant size, and then show that the optimal integral solution gives rise to a seedset of size $\Theta(h)$.

\myparab{The fractional solution.} Table~\ref{table:frac1} describes a feasible
fractional solution of constant size. We now walk through this solution. In
the solution, we group the rows in the following way:
\begin{itemize}
\item The first row corresponds with the root node.
\item The rest of the rows are grouped by ``stripes''.
A stripe consists of a seed candidate and its corresponding tail.
For example, the first stripe consists of the rows for $s_1$, $v_{1, 1}$, ..., $v_{1, w}$.
\end{itemize}

\begin{sidewaystable}
\caption{A fractional solution for the simple IP formulation}
\label{table:frac1}
\begin{tabular}{lll|l||l|lllll|lllll|c||lllll|}
\multicolumn{5}{c}{} & \multicolumn{5}{c}{First cycle} & \multicolumn{5}{c}{Second cycle}& \multicolumn{1}{c}{$h-2$ other cycles}&\multicolumn{5}{c}{Completion stage}\\
\cline{4-21}
& & $R$ & 1 & 0 & 0 &  0 & 0 & \ldots & 0& 0& 0& 0&\ldots &  0 & \ldots & 0 & 0& 0 & \ldots& 0  \\
\cline{4-21}
 \multirow{6}{*}{\rotatebox{270}{1st stripe}
}&
\multirow{6}{*}{$ \hspace{-.5cm}\left\{
\begin{array}{c}
\\ \\ \\ \\ \\ \\
\end{array}
\right.$
} & $s_1$ & 0 & $\epsilon$ & 0 & 0 & 0 & \ldots & 0 & 0&0 &0 &\ldots &0 &\ldots & $1-\epsilon$ & 0 & 0 &\ldots & 0\\
\cline{4-21}
& & $v_{1,1}$ & 0 & 0 & $\underline{\epsilon}$ & 0 & 0   &  \ldots &   0  & $\underline{\epsilon}$ & 0 & 0   &  \ldots &   0 & \ldots & 0 & 0 & 0 & \ldots & 0\\
  & & $v_{1,2}$ &0 & 0 & 0 & $\underline{\epsilon}$ & 0 & \ldots & 0 &  0 & $\underline{\epsilon}$ & 0 & \ldots & 0 &\ldots & 0 & 0 & 0 & \ldots & 0 \\
& & $v_{1,3}$ &0 & 0 & 0 & 0 & $\underline{\epsilon}$ & \ldots & 0 &  0 & 0 & $\underline{\epsilon}$ & \ldots & 0 & \ldots & 0 & 0 & 0 & \ldots & 0 \\
& &\vdots & & & & & & $\ddots$&  & & & & $\ddots$ & &$\ddots$ & & & & \ldots & \\
& &$v_{1,w}$ & 0 & 0 & 0 & 0 & 0 & \ldots & $\underline{\epsilon}$  & 0 & 0 & 0 & \ldots & $\underline{\epsilon}$ & \ldots & 0 & 0 & 0 & \ldots & 0 \\
\cline{4-21}
 \multirow{6}{*}{\rotatebox{270}{2nd stripe}
}&
\multirow{6}{*}{$ \hspace{-.5cm}\left\{
\begin{array}{c}
\\ \\ \\ \\ \\ \\
\end{array}
\right.$
} & $s_2$ & 0 & $\epsilon$ & 0 & 0 & 0 & \ldots & 0 & 0&0 &0 &\ldots &0 &\ldots & $\epsilon$ & $1-2\epsilon$ & 0 &\ldots & 0\\
\cline{4-21}
& & $v_{2,1}$ & 0 & 0 & $\epsilon$ & 0 & 0   &  \ldots &   0  & $\epsilon$ & 0 & 0   &  \ldots &   0 & \ldots & 0 & 0 & 0 & \ldots & 0\\
  & & $v_{2,2}$ &0 & 0 & 0 & $\epsilon$ & 0 & \ldots & 0 &  0 & $\epsilon$ & 0 & \ldots & 0 &\ldots & 0 & 0 & 0 & \ldots & 0 \\
& & $v_{2,3}$ &0 & 0 & 0 & 0 & $\epsilon$ & \ldots & 0 &  0 & 0 & $\epsilon$ & \ldots & 0 & \ldots & 0 & 0 & 0 & \ldots & 0 \\
& &\vdots & & & & & & $\ddots$&  & & & & $\ddots$ & &$\ddots$ & & & & \ldots & \\
& &$v_{2,w}$ & 0 & 0 & 0 & 0 & 0 & \ldots & $\epsilon$  & 0 & 0 & 0 & \ldots & $\epsilon$ & \ldots & 0 & 0 & 0 & \ldots & 0 \\
\cline{4-21}
 \multirow{6}{*}{\rotatebox{270}{3rd stripe}
}&
\multirow{6}{*}{$ \hspace{-.5cm}\left\{
\begin{array}{c}
\\ \\ \\ \\ \\ \\
\end{array}
\right.$
} & $s_3$ & 0 & $\epsilon$ & 0 & 0 & 0 & \ldots & 0 & 0&0 &0 &\ldots &0 &\ldots & $0$ & $2\epsilon$ & $1-3\epsilon$ &\ldots & 0\\
\cline{4-21}
& & $v_{3,1}$ & 0 & 0 & $\epsilon$ & 0 & 0   &  \ldots &   0  & $\epsilon$ & 0 & 0   &  \ldots &   0 & \ldots & 0 & 0 & 0 & \ldots & 0\\
  & & $v_{3,2}$ &0 & 0 & 0 & $\epsilon$ & 0 & \ldots & 0 &  0 & $\epsilon$ & 0 & \ldots & 0 &\ldots & 0 & 0 & 0 & \ldots & 0 \\
& & $v_{3,3}$ &0 & 0 & 0 & 0 & $\epsilon$ & \ldots & 0 &  0 & 0 & $\epsilon$ & \ldots & 0 & \ldots & 0 & 0 & 0 & \ldots & 0 \\
& &\vdots & & & & & & $\ddots$&  & & & & $\ddots$ & &$\ddots$ & & & & \ldots & \\
& &$v_{3,w}$ & 0 & 0 & 0 & 0 & 0 & \ldots & $\epsilon$  & 0 & 0 & 0 & \ldots & $\epsilon$ & \ldots & 0 & 0 & 0 & \ldots & 0 \\
\cline{4-21}
& & \vdots & \vdots & \vdots & & & & \vdots & & & & & \vdots & \vdots & & & & &\vdots &\\
\cline{4-21}
 \multirow{6}{*}{\rotatebox{270}{h-th stripe}
}&
\multirow{6}{*}{$ \hspace{-.5cm}\left\{
\begin{array}{c}
\\ \\ \\ \\ \\ \\
\end{array}
\right.$
} & $s_h$ & 0 & $\epsilon$ & 0 & 0 & 0 & \ldots & 0 & 0&0 &0 &\ldots &0 &\ldots & $0$ & 0 & 0 &\ldots & $(h-1)\epsilon$\\
\cline{4-21}
& & $v_{h,1}$ & 0 & 0 & $\epsilon$ & 0 & 0   &  \ldots &   0  & $\epsilon$ & 0 & 0   &  \ldots &   0 & \ldots & 0 & 0 & 0 & \ldots & 0\\
  & & $v_{h,2}$ &0 & 0 & 0 & $\epsilon$ & 0 & \ldots & 0 &  0 & $\epsilon$ & 0 & \ldots & 0 &\ldots & 0 & 0 & 0 & \ldots & 0 \\
& & $v_{3,3}$ &0 & 0 & 0 & 0 & $\epsilon$ & \ldots & 0 &  0 & 0 & $\epsilon$ & \ldots & 0 & \ldots & 0 & 0 & 0 & \ldots & 0 \\
& &\vdots & & & & & & $\ddots$&  & & & & $\ddots$ & &$\ddots$ & & & & \ldots & \\
& &$v_{h,w}$ & 0 & 0 & 0 & 0 & 0 & \ldots & $\epsilon$  & 0 & 0 & 0 & \ldots & $\epsilon$ & \ldots & 0 & 0 & 0 & \ldots & 0 \\
\cline{4-21}
\end{tabular}
\end{sidewaystable}

We shall also divide the columns into two parts. The first part is the ``false propagation'' stage, consisting of $2+wh$ columns where we use a small fractional seed set to activate the tail nodes. The second part is the ``completion stage'' consisting of $h-1$ rows \Snote{check this! }, where we fill in the residual mass of the nodes so that the permutation constraints are met.

\mypara{Variable assignments in the fractional solution. } We now describe the assignments in Table~\ref{table:frac1}.
\begin{itemize}
\item $x_{R, 1} = 1$, i.e. the root is first activated.
\item Let $\epsilon \triangleq 1/h$. For each stripe $\{s_i, v_{i, 1}, ..., v_{i, w}\}$, we assign values in the false propagation region as follows:
    $x_{s_i, 2} = \epsilon$ and $x_{v_{i,j}, j + 2 + k w} = \epsilon$ for all $j \leq w$ and $0 \leq k < h$. The rest of the variables in this region are set to $0$. This assignment exhibits a periodic pattern, so that mass can circulate back and forth along a tail until all nodes in the tail are activated at the end of the false propagation stage. (Refer to the underlined values in the first stripe of Table~\ref{table:frac1}).
\item \Snote{Zhenming, I found some small bugs here with your column numbers and changed them. Please check them. } Finally, we fill in the variables in the completion stage so that the permutation constraints are met. Notice that at time $n - h + 2$, only the rows that correspond with the seed candidates $s_i$ ($i \leq h$) do not sum up to 1. We use the columns in the completion stage to fill in the extra mass using a ``greedy'' approach. 
    In other words, at the column $n - h + 2$, we first fill in the unused mass (namely $1-\epsilon$) from $s_1$. Then we fill in the unused mass from $s_2$ as much as possible, subject to the constraint that the column sums to 1 (namely $\epsilon$). Next, we move to the next column (the $n - h + 3$-rd column). Then we fill in the mass from $s_2$ and as much mass as possible from $s_3$ to this column. This process continues until all mass from $s_i$ ($i \leq h$) is filled.
\end{itemize}

\mypara{The fractional solution is feasible. }
Next, we argue that such assignments are feasible. Since, we satisfied the permutation constraints by construction (Table~\ref{table:frac1}), we only argue that the connectivity constraints are met.
\begin{itemize}
\item We need to start thinking about connectivity when $t = 2$. At this time step,
the connectivity constraints are met because all the seed candidates are connected to the root $R$, which is activated at time $t = 1$.

\item Next, we argue that the connectivity constraints are met in the propagation stage. Let us consider the first cycle in the propagation stage. In the first time step of the first cycle, an $\epsilon$-fraction of mass is activated at $v_{i,1}$ for all $i \leq h$. Since $v_{i,1}$ is connected with $s_i$, and an $\epsilon$ portion of $s_i$ is active prior to the beginning of the 1st cycle, the connectivity constraint is met for this step. For the rest of the timesteps of the first cycle, note that by the time we assign $\epsilon$ to the node $v_{i,j}$, an $\epsilon$ portion of mass is already activated at $v_{i,j-1}$.   Since $\{v_{i,j-1}, v_{i,j}\}\in E$ for all $j < w$, the connectivity constraints are met for the entire first cycle.
The argument for the remaining cycles proceeds in a similar manner.

\item Finally, showing the connectivity holds in the completion stage is trivial: this follows because only seed candidates activate at the this stage, and seed candidates are all connected to the root which has been fully activated since $t = 1$.
\end{itemize}
\noindent
Hence, we can conclude that the fractional solution in Table~\ref{table:frac1} is feasible.

\myparab{The integral solution.}  To prove that the optimal integral solution is a seedset of size $O(h)$, we show that any seedset of size less than $\frac h 5$ will fail to activate all the nodes in the graph. Here, the constant $\frac 1 5$ is chosen rather arbitrarily and is not optimized.

First, we notice that for any feasible set $S$ that contains one or more tail nodes,
we can transform it into a feasible set $S'$ such that (a) $|S'|\leq |S|$ and (b) no tail nodes are in $S'$.  To construct the new seedset, replace each tail node $v_{i,j}$ in $S$ by it's parent seed candidate $s_i$. Since the activation of $s_i$ always causes the activation of $v_{i,j}$ for any $j$, it follows that $S'$ is a feasible seedset whenever $S$ is a feasible seedset.

Thus, we may focus on the seedset that contains only $R$ and/or seed candidates. Wlog, we may assume the seed set is a subset in $U = \{R, s_1, ..., s_{\frac h 5}\}$. Next, we argue that the seedset $U$ fails to activate all the nodes in the graph.
First, we can see that all the tails $v_{i,j}$ ($i \leq \frac h 5$ and $j \leq w$) with parent seed candidates in $U$ will be active.
After they are activated, the total number of activated nodes will be $\frac h 5 + 1 + \frac{wh} 5$.
Now we argue that no other nodes are active because (a) all seed candidates $s_i$ ($i > \frac h 5$) that are not in $U$ cannot be activated since the following holds
\begin{equation}\label{eq:gap1bound22}
(\tfrac h 5 + 1 + \tfrac{wh} 5) + 1 < \theta(s_i)=n-h+2
\end{equation}
 for sufficiently large constant $w$ and sufficiently large $n$, and (b) all tail nodes $v_{i,j}$ ($i > \frac h 5$ and $j \leq w$) cannot be activated until their parent seed candidate is active.

\myparab{A $\Theta(n)$ integrality gap.} We can conclude that the integral solution has a seedset size of $O(h)$ while the fractional solution is $O(1)$. When we set $w$ be a sufficiently large constant and $h= \Theta(n)$ (we only need to ensure that (\ref{eq:gap1bound22}) holds), our integrality gap is $\Theta(n)$.

\subsubsection{Integrality gap for the augmented IP of Figure~\ref{fig:ipflow}.}\label{sec:gap2}
In this section, we shall prove the following theorem.

\begin{theorem}Consider the augmented linear program of of Figure~\ref{fig:ipflow}.
For any sufficiently large $n$ and any $\ell \leq cn^{1/3}$, where $c$ is a suitable constant, there exists a problem instance with an $\Omega(\ell)$ integrality gap.
\end{theorem}

\myparab{The problem instance.}   To simplify the exposition, we will assume that our problem instance $\{G,\theta\}$ is such that our graph $G$ has $|V|=n$ nodes, where $n - 1$ is a multiple of $\ell$, and the range of $\theta$ is $2\ell+2$ different threshold values.  We shall let $w$ be the integer such that $(w+2)\ell + 1 = n$, and let $\epsilon \triangleq 1/\ell$.   Our graph $G$ is described as follows (See Figure~\ref{fig:ipgap2}):
\begin{itemize}
\item the node set consists of the following:
\begin{itemize}
\item The root vertex $R$.
\item The set of ``seed candidate'' $\{s_1, ..., s_\ell\}$.
\item The set of ``blockers'' $\{b_1, ..., b_\ell\}$.
\item The set of ``tails'' $v_{i, j}$, where $i\leq \ell$ and $j \leq w$.
\end{itemize}
\item The edge set consists of the following three types of edges
\begin{itemize}
\item There is an edge between the root and any seed candidate, i.e. $\{R, s_i\} \in E$
for all $i \leq \ell$.
\item There is an edge between the root and any blocker, i.e. $\{R, b_i\} \in E$ for all $i \leq \ell$.
\item For any $i, j$, we  have $\{s_i, v_{i,j}\} \in E$ and $\{b_i, v_{i,j}\} \in E$.
\end{itemize}
\end{itemize}
In what follows, we shall also refer to the subgraph induced by $s_i$, $b_i$, $v_{i, 1}$, ..., $v_{i, w}$ as the $i$-th gadget of the graph.
We set the threshold function $\theta$ as follows:
\begin{itemize}
\item $\theta(R) = n$.
\item $\theta(s_i) = (w+1)\ell + 3$.
\item $\theta(v_{i,j}) = (i - 1)(w+1) + 3$.
\item $\theta(b_i) = (i - 1)(w+1) + w/\ell + 2$.
\end{itemize}

\begin{figure}
\begin{center}
\includegraphics[width=4in]{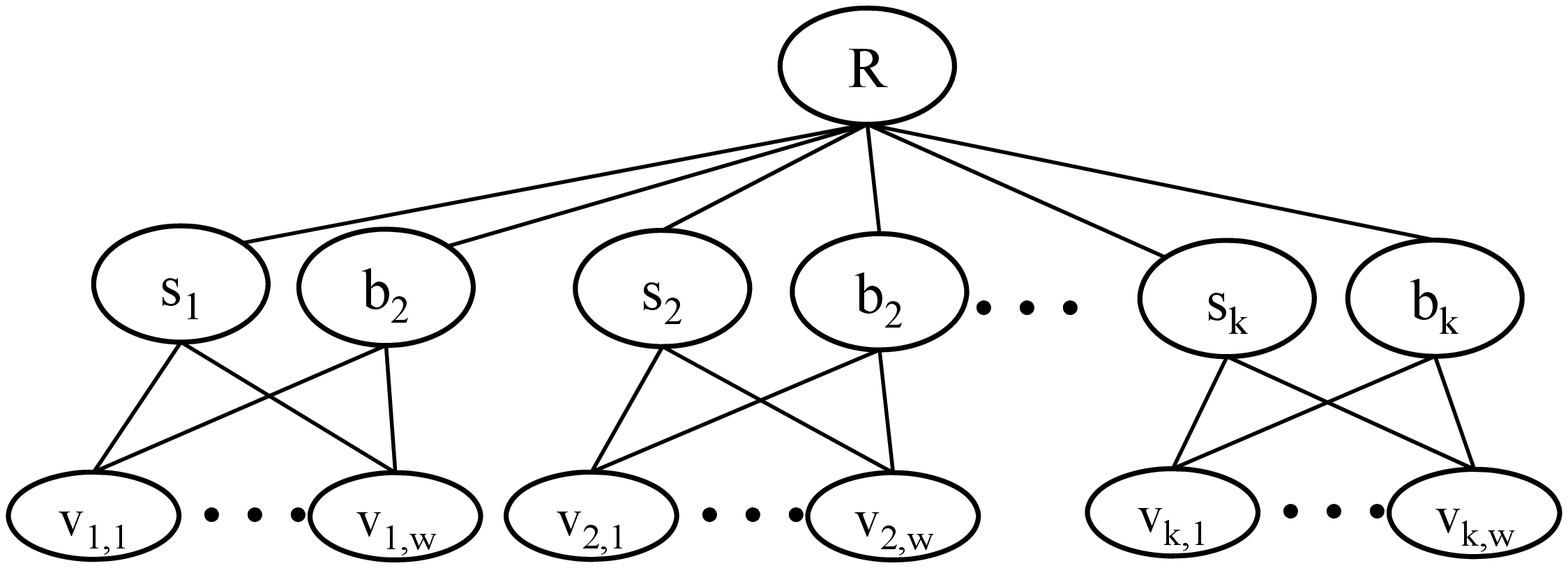}
\caption{The graph for the hard instance with $\Omega(k)$ integrality gap}
\label{fig:ipgap2}
\end{center}
\end{figure}

In what follows, we shall first show that a feasible solution of size $O(1)$ exists for the relaxed LP. Then we shall show that the optimal integral solution is a seedset of size  $\Omega(\ell)$.

\myparab{The fractional solution. } We now construct feasible fractional solution of size $O(1)$.  See the table in Figure~\ref{fig:btable2}.  The intuition behind our construction is to activate the root $R$ at the first time step, i.e. $x_{R, 1} = 1$, and then activate an
$\epsilon$-portion of each seed candidate in the second time step.  This total of $1+\ell\epsilon=2$ mass will be the size of the entire (fractional) seedset.  We will make sure that the rest of the node's mass will activate after their thresholds, and will therefore not contribute to the size of the fraction solution.

We divide our construction into two time stages. The first is the ``false propagation'' stage, where all the nodes except for the seed candidates will be fully activated. The second is the ``completion
stage'' where the remaining inactivated mass from the seed candidates will be activated.
Next, we describe each of these two stages in detail.

\begin{figure}
\begin{tabular}{l|l|l||lll||lll||lll||lll||lllll||}
\cline{4-20}
 \multicolumn{3}{c||}{} & \multicolumn{12}{c||}{Propagation Stage}&  \multicolumn{5}{c||}{Completion Stage} \\
\cline{2-20}
$R$ & 1 & 0 & 0 & \ldots & & 0 & \ldots &  & & \ldots & & 0 & \ldots & & 0 & 0 & 0 & \ldots & 0\\
\cline{2-20}
$s_1$ & 0 & $\epsilon$ & 0 & \ldots & & 0 & \ldots & & & & & 0 & \ldots &  & $1-\epsilon$ & 0 & 0 & \ldots & 0\\
$s_2$ & 0 & $\epsilon$ & 0 & \ldots & & 0 & \ldots & & & & & 0 & \ldots &  & $\epsilon$ & $1-2\epsilon$ & 0 & \ldots & 0\\
\vdots & \vdots & \vdots & & \vdots & &  & \ldots & & & \vdots & &  & \ldots &  &  &  &  & \ldots & \\
$s_\ell$ & 0 & $\epsilon$ & 0 & \ldots & & 0 & \ldots & & & & & 0 & \ldots &  & 0 & 0 & 0 & \ldots & $(\ell-1)\epsilon$\\
\cline{2-20}
$b_1$ & 0 & 0 & \multicolumn{3}{c||}{\multirow{4}{*}{{\huge M}}} & \multicolumn{3}{c||}{\multirow{4}{*}{{\huge 0}}}  & \multicolumn{3}{c||}{\multirow{4}{*}{{\huge \ldots}}} & \multicolumn{3}{c||}{\multirow{4}{*}{{\huge 0}}} & 0 & 0 & 0 & \ldots & 0 \\
$v_{1,1}$ & 0 & 0 & \multicolumn{3}{c||}{} & \multicolumn{3}{c||}{} & \multicolumn{3}{c||}{} & \multicolumn{3}{c||}{} & 0 & 0 & 0 & \ldots & 0\\
\vdots & \vdots & \vdots &  \multicolumn{3}{c||}{} & \multicolumn{3}{c||}{} & \multicolumn{3}{c||}{} & \multicolumn{3}{c||}{} &  &  & & \ldots &  \\
$v_{1,w}$ & 0 & 0 & \multicolumn{3}{c||}{} & \multicolumn{3}{c||}{} & \multicolumn{3}{c||}{} & \multicolumn{3}{c||}{} & 0 & 0 & 0 & \ldots & 0\\
\cline{2-20}
$b_2$ & 0 & 0 & \multicolumn{3}{c||}{\multirow{4}{*}{{\huge 0}}} & \multicolumn{3}{c||}{\multirow{4}{*}{{\huge M}}}& \multicolumn{3}{c||}{\multirow{4}{*}{{\huge \ldots}}} & \multicolumn{3}{c||}{\multirow{4}{*}{{\huge 0}}} & 0 & 0 & 0 & \ldots & 0 \\
$v_{2,1}$ & 0 & 0 & \multicolumn{3}{c||}{} & \multicolumn{3}{c||}{} & \multicolumn{3}{c||}{} & \multicolumn{3}{c||}{} & 0 & 0 & 0 & \ldots & 0\\
\vdots & \vdots & \vdots &  \multicolumn{3}{c||}{} & \multicolumn{3}{c||}{} & \multicolumn{3}{c||}{} & \multicolumn{3}{c||}{} &  &  &  & \ldots & \\
$v_{2,w}$ & 0 & 0 & \multicolumn{3}{c||}{} & \multicolumn{3}{c||}{} & \multicolumn{3}{c||}{} & \multicolumn{3}{c||}{} & 0 & 0 & 0 & \ldots & 0\\
\cline{2-20}
\vdots & \vdots & \vdots & & \vdots & & & \vdots & & & $\ddots$ & & & \ldots & &  &  &  & \ldots &  \\
\cline{2-20}
$b_\ell$ & 0 & 0 & \multicolumn{3}{c||}{\multirow{4}{*}{{\huge 0}}} & \multicolumn{3}{c||}{\multirow{4}{*}{{\huge 0}}} & \multicolumn{3}{c||}{\multirow{4}{*}{{\huge \ldots}}} & \multicolumn{3}{c||}{\multirow{4}{*}{{\huge M}}}& 0 & 0 & 0 & \ldots & 0\\
$v_{\ell,1}$ & 0 & 0 & \multicolumn{3}{c||}{} & \multicolumn{3}{c||}{} & \multicolumn{3}{c||}{} & \multicolumn{3}{c||}{}& 0 & 0 & 0 & \ldots & 0\\
\vdots & \vdots & \vdots &  \multicolumn{3}{c||}{} & \multicolumn{3}{c||}{} & \multicolumn{3}{c||}{} & \multicolumn{3}{c||}{}&  &  &  & \ldots & \\
$v_{\ell,w}$ & 0 & 0 & \multicolumn{3}{c||}{} & \multicolumn{3}{c||}{}& \multicolumn{3}{c||}{} & \multicolumn{3}{c||}{}& 0 & 0 & 0 & \ldots & 0\\
\cline{2-20}
\end{tabular}
\caption{Feasible fractional assignments for the flow based linear program.}
\label{fig:btable2}
\end{figure}

\mypara{False propagation stage.} The false propagation stage consists of $(w+1)\times \ell$ time steps, divided into $\ell$ \emph{blocks}, each consisting of $(w+1)$ time steps. Notice that the thresholds of the seed candidates $\theta(s_i) \;\forall i$ occur exactly after the false propogation stage ends.
%
%
During $i$-th block of the false propagation stage, the blocker $b_i$ and tail nodes $v_{i,j} \;\forall j\in[w]$ in $i$-th gadget will be fully activated. Since there are exactly $(w+1)$ such nodes, the $i$-th block is a $(w+1)\times(w+1)$ matrix, the only non-zero variables in the $i$-th block will be those of the blocker and tail nodes in the $i$-th gadget. These variables are expressed as the sub-matrices $M$ in Figure~\ref{fig:btable2}.

\mypara{The variable assignments in $M$. } We next describe the variable assignments in $M$ for the $i$-th block, as shown in Table~\ref{table:block}. 
Our variable assignments will keep the invariance that before the $i$-th block, all nodes in the $k$-th gadget, \emph{except} for the seed candidate nodes, are fully activated for every $k \leq i - 1$.

The $i$-th block begins at the $3+(i-1)(w+1)$-th time step and ends at the $2+i(w+1)$-th time step.
The assignments in $M$ are divided into multiple \emph{cycles}, each of which
spans $w/\ell=\epsilon w$ time steps. Notice that the $i$-th block will contain in total $\ell$ cycles, and one extra time step that does not belong to any cycle. This extra time step will be inserted between the end of the first cycle and the beginning of the second cycle and will be used to activate the blocker $b_i$, \ie $x_{b_i, 3+(i-1)(w+1)+w/\ell} = 1$.

In each cycle, every node's mass needs to be incremented by $\epsilon$. We do this using a greedy construction, incrementing the mass of $\ell$ tail nodes by $\epsilon$ in each timestep of a given cycle, so that the column constraints are met for this cycle. For this reason, we need $\frac w \ell$ time steps to fully activate all the tail nodes; this follows because there are in total $\ell$ cycles, so the sum of the active portion of any tail node $v_{i,j}$ in the $i$-th block is $\epsilon \cdot \ell = 1$, so that $v_{i,j}$ is completely activated.

\mypara{Feasibility of the assignments.}  We show why our variable assignments at the propagation stage are feasible, and do not increase the mass of the fractional seedset. Our analysis is based on induction. Recall our invariance that prior to the start of the $i$-th block, all blockers and tail nodes in the $j$-th gadget ($k \leq i - 1$) are fully activated. We shall show that if the invariance holds up to the $(i - 1)$st block, the variable assignments in the $i$-th block are feasible (and do not introduce any mass to the seedset).  Suppose that the invariance holds up to the $(i - 1)$-st block. Then we have that:
\begin{itemize}
\item the seedset does not increase, because mass for each node is assigned after its corresponding threshold.  (This follows since the $i$-th block starts at timestep $(w+1)(i-1)+3$ and the tail nodes have $\theta(v_{i,j}) = (w+1)(i-1)+3$.  Similarly, the blocker $b_i$ is activated at time $(w+1)(i-1)+3+w/\ell$ while $\theta(b_i) = (w+1)(i-1)+3+w/\ell$.)
\item the flow constraints are met.  During the first cycle, (timesteps $(w+1)(i-1)+3$ to $(w+1)(i-1)+w/\ell+2$), we may push a flow of size $\epsilon$ to any tail node $v_{i,j}$ through the path $R$-$s_i$-$v_{i,j}$ since the seed candidate $s_i$ has is an $\epsilon$-portion active.
    Next, at timestep $(w+1)(i-1)+w/\ell+3$ the blocker $b_i$ must receive a unit flow; this is feasible since $b_i$ is directly connected to the root that is fully active at $t=1$.
    Finally, during the remaining cycles (timesteps $(w+1)(i-1)+w/\ell+4$ to $(w+1)i + 2$) we may continue to fully activate the tail nodes $v_{i,j}$ by pushing a up to a unit of flow through the path $R$-$b_i$-$v_{i,j}$.
\end{itemize}
Thus, we may conclude that the assignments at the $i$-th block
are feasible and will not increase the size of the seedset, which further implies that the invariance also holds for the $i$-th block.

\mypara{Completion stage.} We now describe the assignments for the completion stage. The completion stage starts at time $t_2 \triangleq 3+\ell(w+1)$. Since $t_2 \geq \theta(s_i)$, activating the seed candidates in this stage does not increase the side of the seedset. Note further that the only rows that do not sum to $1$ correspond to the seed candidates. We again take a greedy approach to fill in the residual mass from the seed candidates (similar to that used in the completion stage of the  integrality gap presented in Section~\ref{sec:gap1}). At column $t_2$, we let $x_{s_1, t_2} = 1-\epsilon$ and $x_{s_2, t_2} = \epsilon$, filling in the unused row mass from $s_1$, and filling in the unused row mass of $s_2$ as much as possible subject to the column constraint of column $t_2$. We repeat this process for each of the $\ell-1$ remaining columns, until all mass from the seed candidates is used up.

By construction, the completion stage satisfies the row and column permutation constraints, and does not increase the size of the fractional seedset (since all seed candidates activate after their thresholds).  Finally, the flow constraints are satisfied since we can push up to a unit of flow to each $s_i$ along its direct connection to $R$.

\begin{sidewaystable}
\caption{Matrix $\mathbf{M}$, the fractional assignments for a block}
\label{table:block}
\begin{tabular}{l|llllll||l||llllll||llllll||lll||llllll|}
\multicolumn{1}{c}{} & \multicolumn{6}{c}{First cycle} & \multicolumn{1}{c}{} & \multicolumn{6}{c}{Second cycle} & \multicolumn{6}{c}{Third cycle}  & \multicolumn{3}{c}{...} & \multicolumn{6}{c}{$\frac w \ell$-th cycle}\\
\cline{2-29}
$b_i$ & 0 & 0 & 0 & \ldots & 0 & 0 & 1 & 0 & 0 & 0 & \ldots & 0 & 0 & 0 & 0 & 0 & \ldots & 0 & 0 & & \ldots & &0 & 0 & 0 & \ldots & 0 & 0\\
\cline{2-29}
$v_{i,1}$ & $\epsilon$ & 0 & 0 & \ldots & 0 & 0 & 0 & $\epsilon$ & 0 & 0 & \ldots & 0 & 0 & $\epsilon$ & 0 & 0 & \ldots & 0 & 0 & & \ldots & & $\epsilon$ & 0 & 0 & \ldots & 0 & 0 \\
$v_{i,2}$ & $\epsilon$ & 0 & 0 & \ldots & 0 & 0 & 0 & $\epsilon$ & 0 & 0 & \ldots & 0 & 0 & $\epsilon$ & 0 & 0 & \ldots & 0 & 0 & & \ldots & & $\epsilon$ & 0 & 0 & \ldots & 0 & 0 \\
\vdots & \vdots & & & \ldots &  & \vdots & \vdots & \vdots & & & \ldots &  & \vdots & \vdots & & & \ldots &  & \vdots & & \ldots & &\vdots & & & \ldots &  & \vdots\\
$v_{i,\ell}$ & $\epsilon$ & 0 & 0 & \ldots & 0 & 0 & 0 & $\epsilon$ & 0 & 0 & \ldots & 0 & 0 & $\epsilon$ & 0 & 0 & \ldots & 0 & 0& & \ldots & & $\epsilon$ & 0 & 0 & \ldots & 0 & 0 \\
\cline{2-29}
$v_{i,\ell+1}$ & 0 & $\epsilon$ & 0 & \ldots & 0 & 0 & 0 & 0 & $\epsilon$ & 0 & \ldots & 0 & 0 & 0 & $\epsilon$ & 0 & \ldots & 0 & 0& & \ldots & & 0 & $\epsilon$ & 0 & \ldots & 0 & 0 \\
$v_{i,\ell+2}$ & 0 & $\epsilon$ & 0 & \ldots & 0 & 0 & 0 & 0 & $\epsilon$ & 0 & \ldots & 0 & 0 & 0 & $\epsilon$ & 0 & \ldots & 0 & 0 & & \ldots & & 0 & $\epsilon$ & 0 & \ldots & 0 & 0 \\
\vdots & \vdots & & & \ldots &  & \vdots & \vdots & \vdots & & & \ldots &  & \vdots& \vdots & & & \ldots &  & \vdots & & \ldots & &\vdots & & & \ldots &  & \vdots\\
$v_{i,2\ell}$ & 0 & $\epsilon$ & 0 & \ldots & 0 & 0 & 0 & 0 & $\epsilon$ & 0 & \ldots & 0 & 0 & 0 & $\epsilon$ & 0 & \ldots & 0 & 0 & & \ldots & & 0 & $\epsilon$ & 0 & \ldots & 0 & 0 \\
\cline{2-29}
\vdots & \vdots & \vdots & & \ldots & & & \vdots  & \vdots & \vdots & & \ldots & & & \vdots & \vdots & & \ldots & & & & \ldots &  & \vdots & \vdots & & \ldots & & \\
\cline{2-29}
$v_{i,w-\ell+1}$ & 0 & 0 & 0 & \ldots & 0 & $\epsilon$ & 0 & 0 & 0 & 0 & \ldots & 0 & $\epsilon$ & 0 & 0 & 0 & \ldots & 0 & $\epsilon$ & & \ldots & & 0 & 0 & 0 & \ldots & 0 & $\epsilon$\\
$v_{i,w-\ell+2}$ & 0 & 0& 0 & \ldots & 0 & $\epsilon$ & 0 & 0 & 0& 0 & \ldots & 0 & $\epsilon$ & 0 & 0 & 0 & \ldots & 0 & $\epsilon$  & & \ldots & & 0 & 0 & 0 & \ldots & 0 & $\epsilon$\\
\vdots & \vdots & & & \ldots &  & \vdots & \vdots  & \vdots & & & \ldots &  & \vdots& \vdots & & & \ldots &  & \vdots & & \ldots & & \vdots & & & \ldots &  & \vdots  \\
$v_{i,w}$ & 0 & 0 & 0 & \ldots & 0 & $\epsilon$ & 0 & 0 & 0 & 0 & \ldots & 0 & $\epsilon$& 0 & 0 & 0 & \ldots & 0 & $\epsilon$  & & \ldots & & 0 & 0 & 0 & \ldots & 0 & $\epsilon$\\
\cline{2-29}
\end{tabular}
\end{sidewaystable}

\myparab{The integral solution. } To prove that the optimal integral solution is a seedset of size $\Omega(\ell)$, we show that any seedset of size $\frac \ell 3$ fails to activate the whole graph when $w \geq ck^2$ for some suitable constant $c$.
Fix an arbitrary seedset $S$. Let
\begin{equation}
\mathcal I \triangleq \left\{ i : \exists v \in S \mbox{ s.t. $v$ is in the $i$-th gadget}\right\}.
\end{equation}
We shall write $\mathcal I = \{i_1, ..., i_q\}$, where $q \leq \frac \ell 3$. Next, let $k$
be the smallest integer that is not in $\mathcal I$.
We proceed to  construct a superset $S'$ of $S$ and argue that $S'$ still fails to activate the whole graph. The set $S'$ is constructed as follows:
\begin{itemize}
\item Any nodes that are in $S$ are also in $S'$.
\item The root $R$ is in $S'$.
\item Any tail or blocker nodes that are in the first $(k-1)$-st gadgets are in $S'$.
\end{itemize}
We argue that $S'$ will not activate any additional nodes in the graph. For the sake of contradiction, suppose $u \notin S'$ is the first node activated when $S'$ is the seedset. There are two cases:

\mypara{Case 1.} $u$ is in the $k$-th gadget. The topology of the graph $G$ ensures that any tail node $v_{k, j}$ cannot be activated before $s_{k}$ or $b_{k}$. Therefore, $u$ cannot be $v_{k, j}$. One can see that $\min\{\theta(b_{k}), \theta(s_{k})\} = \theta(b_{k}) = (k - 1)(w+1)+3+\epsilon w$.
On the other hand, we only have
$$|S'| \leq (k-1)(w+1)+\frac \ell 3 + 1.$$
active nodes.  Therefore, when $\frac \ell 3 \leq \epsilon w$ so that $\ell < cn^{1/3}$ for a sufficiently small constant $c$, neither $b_{k}$ nor $s_{k}$ can be active, and so we have a contradiction.

\mypara{Case 2.} $u$ is not in the $k$-th gadget. In this case $\theta(u) \geq k(w+1)+3$. On the other hand, $|S'| \leq (k-1)(w+1)+\frac \ell 3 +1$, so that when $\ell < cn^{1/3}$ for sufficiently large $\ell$, the total number of active nodes is less than $\theta(u)$, which is also a contradiction.

\myparab{The integrality gap.} The optimal fractional solution has size $O(1)$, while the optimal integral solution is a seedset of size $\Omega(\ell)$ (when $\ell < cn^{1/3}$ for large enough constant $c$), so our integrality gap is $\Omega(\ell)$.

\subsection{Remark on the role of flow constraints in reducing the integrality gap}

Finally, we remark the role of flow constraints in reducing the integrality gap from $O(n)$ to $O(\ell)$. From the two gap instances we presented, we can see that there are two types of ``bad'' mass that can adversarially impact the quality of the linear program:
\begin{enumerate}
\item The recirculation of ``fake'' mass, as discussed in the pathological example of Section~\ref{sec:badexample}.  We used recirculation of mass to construct the gap instance for the simple IP of Figure~\ref{fig:ip1} in Appendix~\ref{sec:gap1}.

\item A chain of fractional mass. Recall that both our gap instances (Appendix~\ref{sec:gap1} and Appendix~\ref{sec:gap2}), used a seed candidate $s_i$ to connect to a set of $w=1/\epsilon$ tail nodes $v_{i,1}, ..., v_{i,w}$ so that when an $\epsilon$-portion of $s_i$ becomes active, the total active fractional mass is $\epsilon \cdot (w+1) > 1$. Meanwhile, in the integral solution, we need to activate at least one seed to have a full unit of active mass, which creates a gap of size $1/\epsilon$.
\end{enumerate}
The flow constraints eliminate ``bad'' mass of the first type (see Section~\ref{sec:flowIP}), but cannot eliminate the second type.  It turns out that if we only have the second type of ``bad'' mass, the integrality gap becomes $O(\ell)$ instead of $O(n)$.

For ease of exposition, we explain the relationship between the gap and $\ell$ by refering to the problem instance presented in Appendix~\ref{sec:gap2}.  These arguments can also be generalized to other problem instances. Our crucial observation is that the blockers in each of the gadgets have different thresholds. To see why, suppose that two or more gadgets had blockers that did share the same threshold.  Observe that if we add a seed candidate from one of these gadgets to the seedset, all the nodes in all these gadgets will become active (because the blockers all have the same threshold). This means that we need to include fewer nodes in seedset for the optimal integral solution, which reduces the size of the integrality gap. To sum up, the idea behind our gap instance is to to pad $k$ parallel gadgets together to get a gap of size $\Theta(k)$; for this padding to work we need at least $\Theta(k)$ different threshold values, and so the granularity of the threshold function $\ell$ scales linearly with the integrality gap.

\section{Supplement: Our problem is neither submodular nor supermodular}\label{sec:notSubmod}

We wondered about the relationship between the algorithmic properties of our model and the linear threshold model on social networks articulated in~\cite{KKT}.
\cite{Chen08} showed that the problem of selecting an optimal seedset in the linear threshold mode in social networks cannot be approximated within a factor of $O(2^{\log^{1-\epsilon}|V|})$ when the thresholds are deterministic and known to the algorithm. \cite{KKT} got around this lower bound by assuming that nodes' thresholds are chosen uniformly at random \emph{after} the seedset is selected, and designing an algorithm that chooses the optimal seedset \emph{in expectation}. Their $(1-1/e-\epsilon)$-approximation algorithm relies on the submodularity of the \emph{influence function}, \ie the function $f(S)$ which gives the expected number of nodes that activate given that nodes in $S$ are active.

In this section, we shall show that algorithmic results for submodular and/or supermodular optimization do \emph{not} directly apply to our problem, even if we restrict ourselves to (a) graphs of constant diameter, (b) diffusion problems with a small number of fixed thresholds, or if (c) we choose the thresholds uniformly at random as in \cite{KKT}. Moreover, we see neither diminishing, nor increasing marginal returns even if we restrict ourselves to (d) connected seedsets.

\subsection{Fixed threshold case}

In this section, we construct two families of technology diffusion instances where the threshold function $\theta$ is given as input.  Each family will be on a graph of diameter at most $4$, and require at most $2$ different threshold values, and each will consider connected seedsets.  The first family will \emph{fail} to exhibit the submodularity property while the second will \emph{fail} to exhibit supermodularity.

Let $\{G, \theta\}$ be an arbitrary technology diffusion problem. We shall write $f_{G, \theta}(S)$ be the total number of nodes that eventually activate after seedset $S$ activates. When $G$ and $\theta$ are clear from the context, we simply refer to $f_{G, \theta}(S)$ as $f(S)$.

\subsubsection{The influence function is not submodular.}

\begin{figure}\begin{center}
\includegraphics[width=3.5in]{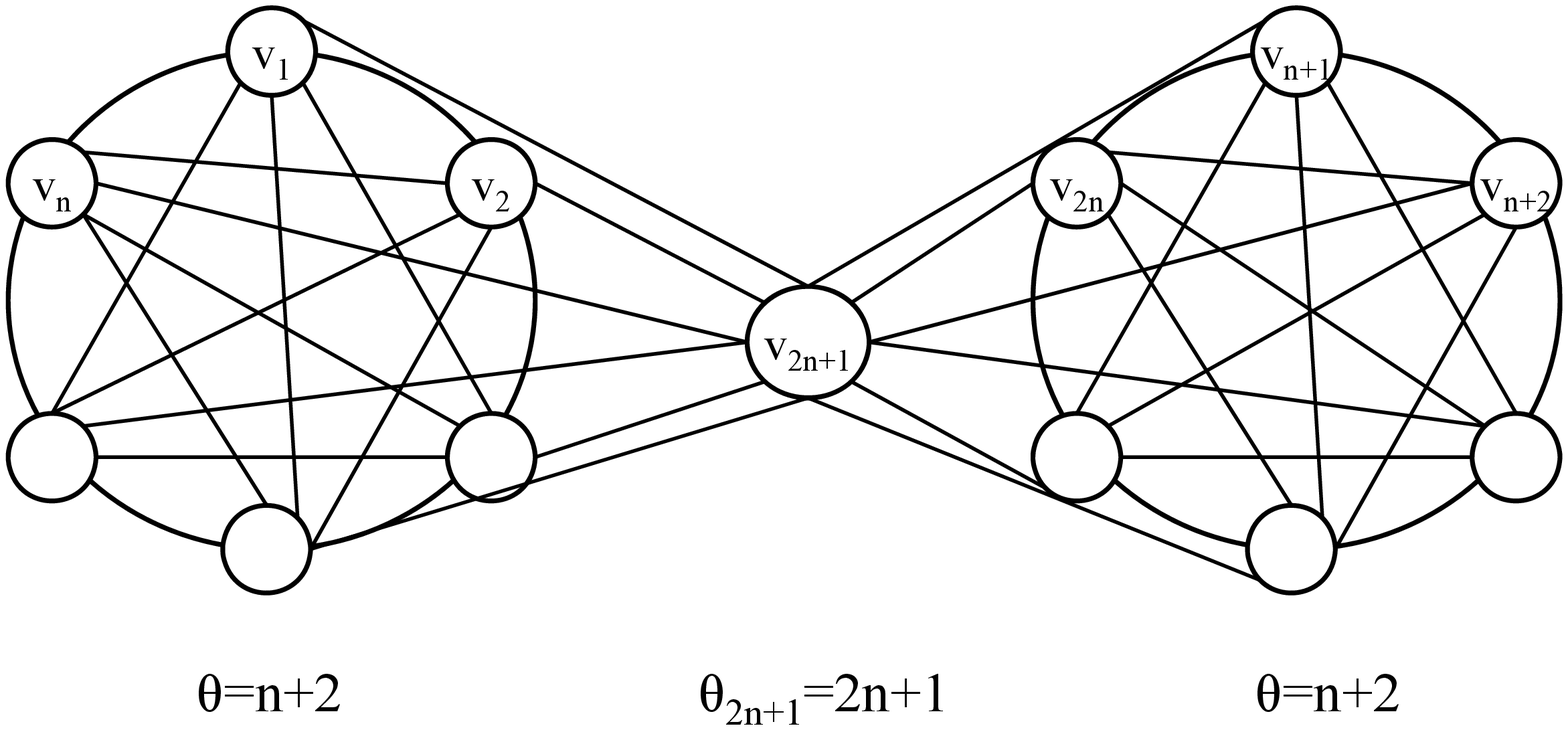}
\caption{An instance of the technology diffusion problem.}\label{fig:sup2}
\end{center}
\end{figure}

Let $n$ be a sufficiently large integer such that the number of nodes in the graph is $2n+1$. This family of technology diffusion problems (which again is implicitly parameterized by $n$) is shown in Figure~\ref{fig:sup2} and defined as follows:
\begin{itemize}
\item The node set is $\{v_1, v_2, ..., v_{2n +1}\}$.
\item The edge set is constructed as follows,
\begin{itemize}
\item The subsets $\{v_1, ..., v_n\}$ and $\{v_{n + 1}, ..., v_{2n}\}$ form two cliques.
\item Node $v_{2n+1}$ is connected to all other nodes in the graph, \ie edges  $\{v_1,v_{2n+1}\}, ... ,\{v_{2n},v_{2n+1}\}$.
\end{itemize}
\item The threshold function is
\begin{itemize}
\item for $i \leq 2n$, $\theta(v_i) = n + 2$.
\item $\theta(v_{2n + 1}) = 2n + 1$.
\end{itemize}
\end{itemize}
To show this problem is non-submodular, we shall find two disjoint sets $S_1$ and $S_2$ such that
\begin{equation}\label{eqn:dnonsub}
f(S_1) + f(S_2) <f(S_1 \cup S_2)
\end{equation}
We chose $S_1 = \{v_1, ..., v_n\}$ and $S_2 = \{v_{2n+1}\}$.  Note that $S_1$ and $S_2$ are connected, and that $f(S_1) = n$, $f(S_2) = 1$, while $f(S_1 \cup S_2) = 2n+1$ so that (\ref{eqn:dnonsub}) holds. \qed

\subsubsection{The influence function is not supermodular.}

\begin{figure}\begin{center}
\includegraphics[width=2.5in]{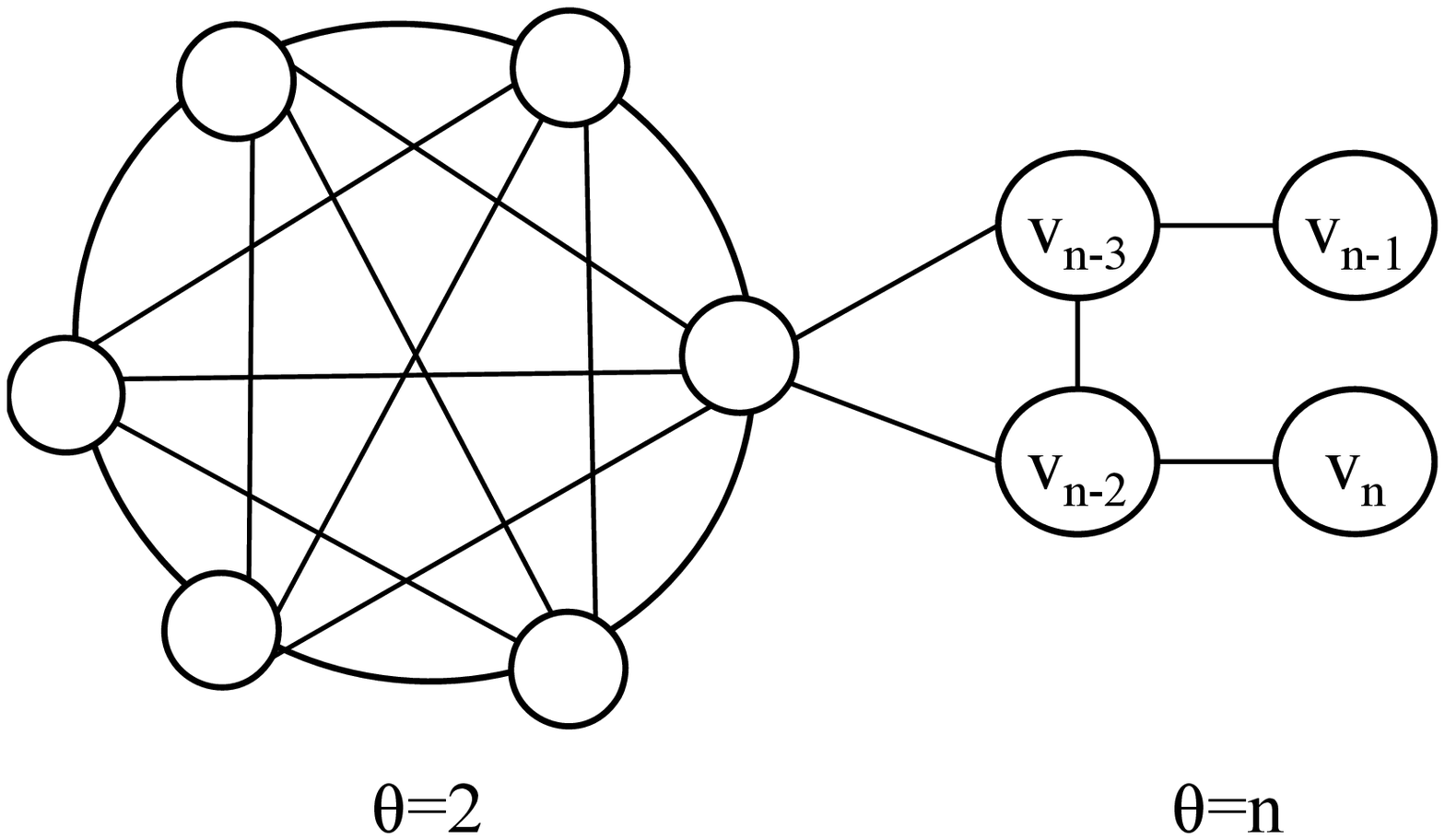}
\caption{Another instance of the technology diffusion problem.}\label{fig:sup1}
\end{center}
\end{figure}

Let $n$ be a sufficiently large integer that represents the
number of nodes in the graph. Our family of technology diffusion problems ${G, \theta}$ (implicitly parameterized by $n$) shown in Figure~\ref{fig:sup1} and defined as follows:
\begin{itemize}
\item The node set is $\{v_1, ..., v_n\}$.
\item The edge set is defined as follows:
\begin{itemize}
\item For any $1 \leq i < j \leq n - 4$, $\{v_i, v_j\}$ is in the edge set, \ie the subgraph induced by
$\{v_1, ..., v_{n - 4}\}$ is a complete graph.
\item The remaining edges are $\{v_1, v_{n - 3}\}$, $\{v_1, v_{n - 2}\}$, $\{v_{n - 3}, v_{n - 1}\}$, $\{v_{n - 2}, v_n\}$, and $\{v_{n - 3},v_{n - 2}\}$.
\end{itemize}
\item The threshold function is
\begin{itemize}
\item For $i \leq n - 4$, $\theta(v_i) = 2$.
\item For $i > n - 4$, $\theta(v_i) = n$.
\end{itemize}
\end{itemize}
To show this problem is not supermodular, we choose two disjoint sets $S_1$ and $S_2$ such that
\begin{equation}\label{eqn:dnonsup}
f(S_1) + f(S_2) > f(S_1 \cup S_2)
\end{equation}
We choose $S_1=\{v_{n - 3}\}$ and $S_2 = \{v_{n - 2}\}$. Note that $S_1$ and $S_2$ are connected, and $f(S_1) = f(S_2) = n - 3$, while $f(S_1 \cup S_2) = n -2$ so that (\ref{eqn:dnonsup}) indeed holds. \qed

\subsection{Randomized threshold case}

We now consider a modified version of our problem, where, as in \cite{KKT}, we assume that thresholds are chosen uniformly at random:
\begin{definition}[Randomized technology diffusion optimization problem.]  The randomized technology diffusion model is as before, with the exception that nodes choose their thresholds uniformly and independently at random from the set $\{2,3,...,n\}$.
Thus, the randomized technology diffusion optimization problem is to find the smallest feasible seedset $S$ \emph{in expectation over the choice of thresholds}, when $G$ is given as input.
\end{definition}
We follow \cite{KKT} and let the influence function $f_G(S)$ be the \emph{expected} number of nodes that are eventually activated, \ie $f_G(S) = \E_{\theta}[f_{G, \theta}(S)]$, where $f_{G, \theta}(S)$ is the number of activated nodes, and expectation is taken over the choice of thresholds. We present two families of problem instances: each family will be on a graph of diameter at most $4$, and will consider connected seedsets. The first family will \emph{fail} to exhibit submodularity of $f_G(S)$,  while the second will \emph{fail} to exhibit supermodularity.


\subsubsection{The influence function is not submodular.}

Let $n$ be a sufficiently large integer such that the number of nodes in the network is $2n +1$. Our family of $G$ (parameterized by $n$) is defined as
\begin{itemize}
\item The node set is $\{v_1, v_2, ..., v_{2n +1}\}$.
\item The edge set is constructed as follows,
\begin{itemize}
\item The subsets $\{v_1, ..., v_n\}$ and $\{v_{n + 1}, ..., v_{2n}\}$ form two cliques.
\item The remaining edges are $\{v_{2n + 1}, v_1\}$ and $\{v_{2n + 1}, v_{2n}\}$. 
\end{itemize}
\end{itemize}
Notice that this family of graphs is \emph{almost} identical to the non-submodular example presented in the previous section, shown in Figure~\ref{fig:sup2}, except that now, the middle node $v_{2n+1}$ is only connected to $v_1$ and $v_{2n}$. We shall find two disjoint set $S_1$ and $S_2$ such that
\begin{equation}\label{eqn:rnonsub}
f_G(S_1) + f_G(S_2) < f_G(S_1 \cup S_2).
\end{equation}
Our choice of $S_1$ and $S_2$ is $S_1 = \{v_1, ..., v_{n}\}$ and $S_2 = \{v_{2n+1}\}$.
We start with computing $f_G(S_1)$:
{\small
\begin{equation}\label{eqn:fg}
f_G(S_1) = 
\E[f_{G, \theta}(S_1)\mid \theta(v_{2n+1})\leq n+1]\Pr[ \theta(v_{2n+1})\leq n+1] + \E[f_{G, \theta}(S_1)\mid \theta(v_{2n+1})>n+1]\Pr[ \theta(v_{2n+1})>n+1]
\end{equation}
}
Notice that
\begin{align}
\E[f_{G, \theta}(S_1) \mid \theta(v_{2n+1}) \leq n+1]
 &= \E[f_{G, \theta}(S_1\cup S_2)] = f_G(S_1\cup S_2)\label{eqn:fgs2}\\
\E[f_{G, \theta}(S_1)\mid \theta(v_{2n+1})>n+1] &= n \notag
\end{align}
Therefore, we may rewrite (\ref{eqn:fg}) as
\begin{equation}\label{eqn:fgrec}
f_G(S_1) = f_G(S_1\cup S_2)\Pr[\theta(v_{2n+1})\leq n+1]+ n \Pr[\theta(v_{2n+1})>n+1] = \frac{f_G(S_1\cup S_2)}{2}+\frac n 2.
\end{equation}
We next move to compute $f_G(S_2)$.  To understand how the influence of $S_2=\{v_{2n+1}\}$ spreads, we condition on the thresholds of its neighbors: $\theta(v_1),\theta(v_{2n})$.
\begin{align}
f_G(S_2) &\leq 1\cdot\Pr[\theta(v_{1})>2 \cap \theta(v_{2n})>2]
+ (2n+1)\cdot\Pr[\theta(v_{1})=2 \cup \theta(v_{2n})=2]\notag\\
&=1 (1- \tfrac1{2n})(1- \tfrac1{2n})+ (2n+1)(2\tfrac1{2n}(1-\tfrac1{2n})+\tfrac1{2n}\tfrac1{2n})\notag\\
&=1+2n\tfrac1{2n}(2(1-\tfrac1{2n})+\tfrac1{2n})
\leq 3 \label{eq:fgs2222}
\end{align}
Therefore, from
(\ref{eqn:fgrec}) and (\ref{eq:fgs2222}) we have
\begin{equation}\label{eq:123}
f_G(S_1) + f_G(S_2) \leq 3 + \tfrac12 (f_G(S_1 \cup S_2)+n)
\end{equation}
Recall that our goal is to show that  $f_G(S_1)+f_G(S_2) < f_G(S_1\cup S_2)$.  Using (\ref{eq:123}), we now see that it suffices to prove that
$$f_G(S_1 \cup S_2)> n+6$$
We prove this by conditioning on the event that $S_1 \cup S_2$ activates node $v_{2n}$:
\begin{align*}
f_G(S_1 \cup S_2)
&=
f_G(S_1 \cup S_2 \cup \{v_{2n}\})\Pr[\theta(v_{2n})\leq n+2]+(n+1)\Pr[\theta(v_{2n})>n+2]\\
&\geq (n+2+\tfrac {n-1}2) \tfrac{n+1}{2n} +(n+1)\tfrac{n-1}{2n}\\
&=n+1+\tfrac{n+1}4
\end{align*}
where the first inequality follows because the thresholds of half of the nonseed nodes $\{v_{n+1}, ..., v_{2n-1}\}$ are $\leq n + 1$ in expectation. Thus, we indeed have that $S_1$ and $S_2$ are connected and $f_G(S_1)+f_G(S_2) < f_G(S_1\cup S_2)$ when $n$ is sufficiently large.
\qed


\subsubsection{The influence function is not supermodular.}

Let $n$ be a sufficiently large integer such that the number of nodes in the network is $2n +1$.
Our family of graphs (parameterized by $n$) is defined as follows,

\begin{itemize}
\item The node set is $\{v_1, v_2, ..., v_{2n +1}\}$.
\item The edge set is constructed as follows,
\begin{itemize}
\item The subsets $\{v_1, ..., v_n\}$ and $\{v_{n + 1}, ..., v_{2n}\}$ form two cliques.
\item Node $v_{2n+1}$ is connected to all other nodes in the graph.
\item There is an additional edge $\{v_1,v_{2n}\}$.
\end{itemize}
\end{itemize}

Notice that this family of graphs is almost identical to the one shown in Figure~\ref{fig:sup2}, except for the addition of a single edge $\{v_1,v_{2n}\}$. We shall find two disjoint set $S_1$ and $S_2$ such that
\begin{equation}\label{eqn:rnonsub}
f_G(S_1) + f_G(S_2) > f_G(S_1 \cup S_2).
\end{equation}
Our choice of $S_1$ and $S_2$ is $S_1 = \{v_1, ..., v_{n}\}$ and $S_2 = \{v_{n+1},...,v_{2n}\}$. Notice that these sets are connected by the edge $\{v_1,v_{2n}\}$.
By symmetry we have that $f(S_1)=f(S_2)$, so we start by computing $f_G(S_1)$.  Let $T$ be the number of active nodes in $S_2$, and let $A$ be the event that node $v_{2n+1}$ is active.
\begin{align}
\E[f_{G, \theta}(S_1)]
& \geq n + (1+\E[T |A, S_1 \text{ active}])\Pr[A|S_1\text{ active}]\notag\\
& \geq n + (1 + n\cdot\tfrac{n+1}{2n}) \tfrac{n}{2n}\notag\\
& = n + \tfrac12(1+\tfrac{n+1}4)\label{eq:supXYZ}
\end{align}
where the second inequality follows because we used the trivial bound  $E[T |A, S_1 \text{ active}] \geq n\tfrac{n+1}{2n}$ where we ignore all cascading effects; we simply assume that each of the $n$ nodes in $S_2$ is connected to an active component of size $n+1$.
On the other hand,
\begin{equation}
\E[f_{G, \theta}(S_1 \cup S_2)]
 \leq 2n + 1 \label{eq:supWYZ}
\end{equation}
Thus we indeed have $f_G(S_1) + f(S_2) \geq 2n + 1+\tfrac{n+1}4  >  2n+1 =f_G(S_1\cup S_2) $ for all $n$.
\qed

\section{\ifnum\conference=0 Supplement:\fi Expository examples and figures}\label{apx:example}

\ifnum \conference=0
\begin{figure}[h]
\else
\begin{figure*}[htp]
\fi
\centering
\subfigure{
\begin{pspicture}(5,5)
\cnodeput(0,2.5){A}{A}
\cnodeput(-1,1.5){B}{B}
\cnodeput(2,0.5){C}{C}
\cnodeput(3, 2){D}{D}
\cnodeput(0, -0.5){F}{F}
\cnodeput(2, -1){E}{E}
\ncline{A}{B}
\ncline{B}{C}
\ncline{A}{C}
\ncline{C}{D}
\ncline{A}{D}
\ncline{C}{F}
\ncline{C}{E}
\ncline{E}{F}
\ncline{B}{E}
\end{pspicture}
}
\subfigure{
\begin{tabular}{c}
\hline
Threshold function  \\
\hline
$\theta(A) = 5$ \\
$\theta(B) = 2$ \\
$\theta(C) = 3$ \\
$\theta(D) = 5$ \\
$\theta(E) = 4$ \\
$\theta(F) = 6$ \\
\hline
\end{tabular}
}
\hspace{5mm}
\subfigure
{
\begin{tabular}{c}
\hline
A connected activation sequence \\
\hline
$x_{A, 1} = 1$, $\left(\forall t \neq 1, x_{A, t} = 0\right)$\\
$x_{B, 2} = 1$, $\left(\forall t \neq 2, x_{B, t} = 0\right)$ \\
$x_{C, 3} = 1$, $\left(\forall t \neq 3, x_{B, t} = 0\right)$ \\
$x_{D, 5} = 1$, $\left(\forall t \neq 5, x_{B, t} = 0\right)$ \\
$x_{E, 6} = 1$, $\left(\forall t \neq 6, x_{B, t} = 0\right)$ \\
$x_{F, 4} = 1$, $\left(\forall t \neq 4, x_{B, t} = 0\right)$ \\
\hline
\end{tabular}
}
\caption{A problem instance and a feasible connected activation sequence. }
\label{fig:diffuseexample}
\ifnum\conference=0
\end{figure}
\else
\end{figure*}
\fi

\ifnum\conference=0
\begin{figure}[h]
\else
\begin{figure*}[b]
\fi
\centering
\includegraphics[width=4in]{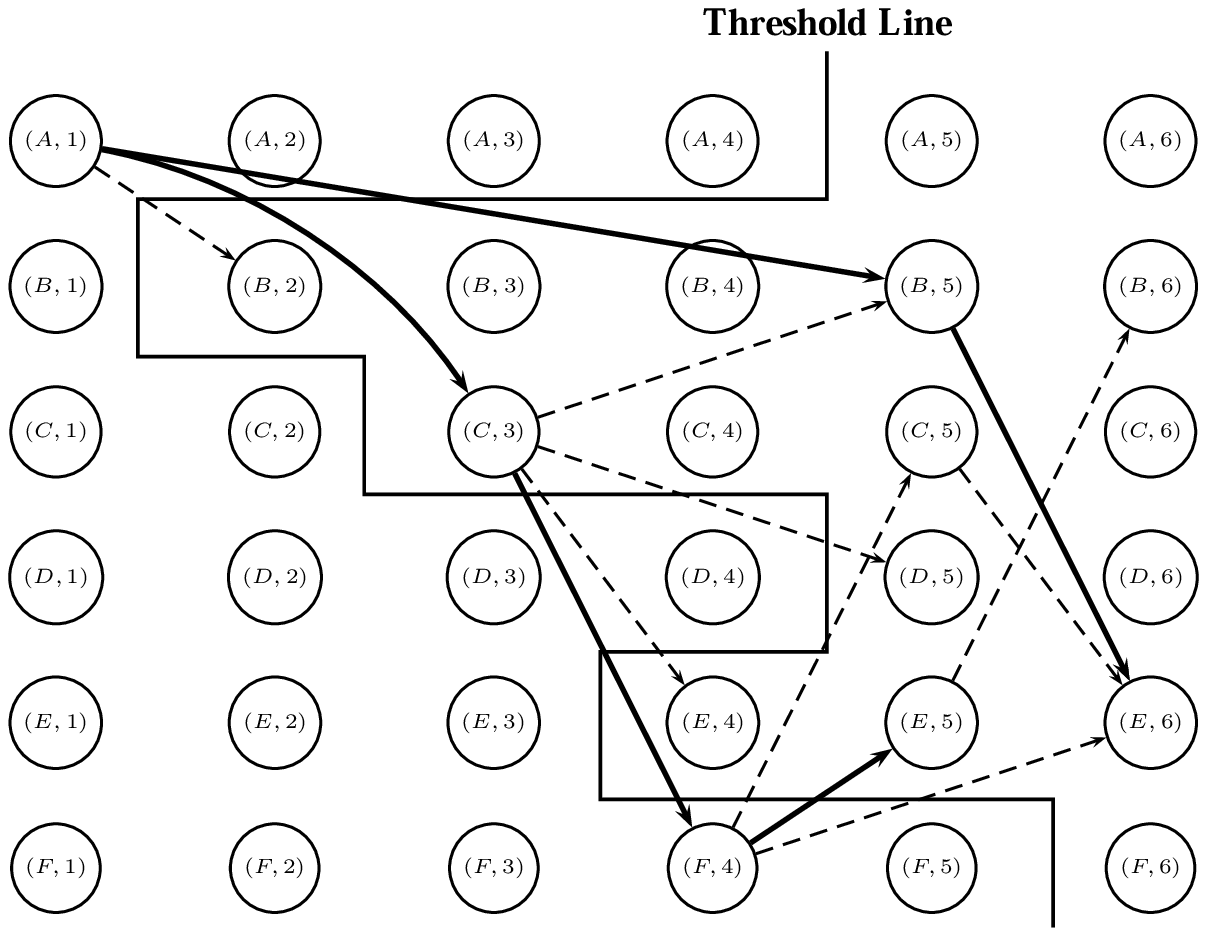}
\caption{The $\mathcal H$ graph and the trajectories of flows.}
\label{fig:hgraph2}
\ifnum\conference=0
\end{figure}
\else
\end{figure*}
\fi

We now present examples of the constructions we used in Section~\ref{sec:linearize} and Section~\ref{sec:roundingAlgo}.
We start with the problem instance $\{G,\theta\}$ in Figure~\ref{fig:diffuseexample}, and present a feasible connected activation sequence $T$ for this problem instance. This feasible connected activation sequence $T$ uniquely corresponds to the seedset $S=\{A,D\}$, since these are the only nodes that have $T(u)<\theta(u)$.

The flow graph $\mathcal H$ used for the relaxed linear program is shown in  Figure~\ref{fig:hgraph2}.  The solid line is the \emph{threshold line}. The (solid and dotted) trajectories represent some paths that can be used to push some amount of flow $f\in[0,1]$ between the nodes in $\mathcal H$ (\ie so that the flow constraints are satisfied).
Notice that every trajectory in $\mathcal H$ corresponds to an edge in the original graph $G$.
Let us consider the $(E, 6)$-flow problem. The solid trajectories in Figure~\ref{fig:hgraph2} illustrate a feasible flow to solve the problem, which we use as the representative flow $\mathcal F_{E,6}$. Notice that $(E,6)$-flow has demand from two nodes $(E, 5)$ and $(E, 6)$ and thus
$\mathcal F_{E,6}$ has two sinks. We decompose $\mathcal F_{E,6}$ into two paths $\mathcal P_1 = (A,1),(C,3),(F,4),(E,5)$ and
$\mathcal P_2 = (A, 1), (B, 5), (E, 6)$.
The \emph{border node} for path $\mathcal P_1$ is $\border(\mathcal P)= (F,4)$ and
the border node for path $\mathcal P_2$ is $\border(\mathcal P)= (A,1)$. Thus,
$\beta(E, 6) = \{(A, 1), (F,4)\}$ and $B(E, 6) = \{A, F\}$.

As an example of \proc{Get-Seq}, suppose that the seedset is $S=\{A,C,F\}$ (note that $S$ is connected), and refer to the graph in Figure~\ref{fig:diffuseexample}.  First, we would ``activate'' vertices in $\mathcal H$ that correspond to seedset nodes, namely $X_{A,1},...,X_{A,6}$, and $X_{C,1},...,X_{C,6}$, and $X_{F,1}, ..., X_{F,6}$.  Next, we iterate over
each of the timesteps. At $t = 2$, since the $\{A, B\} \in  E(G)$ and thus $(X_{A, 1}, X_{B, 2})$ is in $E(\mathcal H)$, we may activate $X_{B, 2}$ as well as all $X_{B, t}$ for all $t > 2$.
At time $t = 3$, we do not find any new nodes that can be activated.
Similarly, at time $t = 4$, we can activate $X_{E, 4}$-$X_{E, 6}$ since $\{E, C\}$ is an edge; at time $t = 5$, we can activate $X_{D, 5}$-$X_{D,6}$. Finally at $t = 6$, all nodes are activated and nothing needs to be done at this step. So we finally obtain the activation sequence:
$$T=(\{A,C,F\}, \{B\}, \bot, \{E\}, \{D\}, \bot ).$$

\end{document}